\documentclass[
a4paper,
11pt,
]
{article}

\usepackage{enumerate}
\usepackage{amsmath}  % Mathe
\usepackage{amssymb} % Mathe Symbole
\usepackage{amsthm} % Theorem-Umgebung
\usepackage{mathtools} % insb fuer das definiert-als-Zeichen (":=") <- \coloneqq 
\usepackage{ulem} % verschiedene Arten zum Unterstreichen
\usepackage{color}
\usepackage{algorithm}% http://ctan.org/pkg/algorithms
\usepackage{algpseudocode}% http://ctan.org/pkg/algorithmicx
\usepackage{algorithmicx}
\usepackage[usenames,dvipsnames]{pstricks}
\usepackage{pst-node}
\usepackage{fullpage}
\usepackage{authblk}
\usepackage{enumitem}						
\usepackage{cases}							
\usepackage[colorlinks=true,citecolor=black,linkcolor=black]{hyperref}

\theoremstyle{plain}	% three theorem styles provided: plain, definition and remark
\newtheorem{definition}{Definition}
\newtheorem{theorem}{Theorem} 
\newtheorem{proposition}{Proposition} 
\newtheorem{lemma}{Lemma} 
\newtheorem{corollary}{Corollary} 
\newtheorem{case}{Case} 
\theoremstyle{remark}
\newtheorem{observation}{Observation}

\newcommand{\vw}{\omega}
\newcommand{\N}{{\mathbb N}}

\newcommand{\vectorname}{mds}

\newcommand{\IN}{\ensuremath{\mathbb{N}}}   % natürliche Zahlen
\newcommand{\IZ}{\ensuremath{\mathbb{Z}}}   % ganze Zahlen
\newcommand{\disjcup}{\mathbin{\dot{\cup}}}	%disjoint union

\let\emph\textbf	\newcommand{\mathemph}[1]{\boldsymbol{#1}}

\title{The Parameterized Complexity of the Equidomination Problem}

\author[1]{Oliver Schaudt\thanks{schaudt@mathc.rwth-aachen.de}}
\author[2]{Fabian Senger\thanks{senger@zpr.uni-koeln.de}}
\affil[1]{RWTH Aachen University, Chair for Mathematics C}
\affil[2]{University of Cologne, Department of Computer Science}

\begin{document} 

\maketitle

\begin{abstract}
A graph $G=(V,E)$ is called equidominating if there exists a value $t\in\IN$ and a weight function $\vw \colon V \rightarrow \IN$ such that the total weight of a subset $D\subseteq V$ is equal to $t$ if and only if $D$ is a minimal dominating set. 
To decide whether or not a given graph is equidominating is referred to as the \textsc{Equidomination} problem.

In this paper we show that two parameterized versions of the \textsc{Equidomination} problem are fixed-parameter tractable:
the first parameterization considers the target value $t$ leading to the \textsc{Target}-$t$ \textsc{Equidomination} problem.
The second parameterization allows only weights up to a value $k$, which yields the $k$-\textsc{Equidomination} problem.

In addition, we characterize the graphs whose every induced subgraph is equidominating.
We give a finite forbidden induced subgraph characterization and derive a fast recognition algorithm.\\[1em]
\noindent \textbf{Keywords:}
minimal dominating set \textperiodcentered\
equidominating graph \textperiodcentered\
kernelization \textperiodcentered\
parameterized complexity \textperiodcentered\
hereditary property
\end{abstract}

\section{Introduction}

Let $G$ be a simple, undirected graph.
A subset $S$ of the vertices of $G$ is called a \emph{dominating set}, if every vertex of $G$ is an element of $S$ or adjacent to a vertex of $S$.
If a dominating set does not contain another dominating set as a subset, it is called a \emph{minimal dominating set}.
Throughout this paper we use the abbreviation \emph{mds} for minimal dominating sets.

While the main stream of the research on dominating sets in graphs focuses on the optimization aspects of the problem, there are several interesting graph classes defined around this concept, for example the class of efficient dominating graphs~\cite{Brandstädt2012}, of well-dominated graphs~\cite{finbow1988well}, of domination perfect graphs~\cite{JGT3190150202}, of upper domination perfect~\cite{GUTIN199895} graphs and of strong domination perfect~\cite{RAUTENBACH2001297} graphs.

Another example is the class of \emph{domishold} graphs, introduced by Benzaken and Hammer in~\cite{Benzaken19781}.
These are the graphs for which there are positive weights associated to the vertices of the graph such that a subset $D$ of vertices is dominating if and only if the sum of the weights of the vertices of $D$ exceeds a certain threshold $t$.
In other words, the characteristic vectors of the dominating sets are exactly the zero-one solutions of a linear inequality, where the coefficients of the inequality correspond to the weights of the vertices.

This concept motivated Payan to define \emph{equidominating} graphs \cite{Payan1980}.
Loose\-ly speaking, these are the graphs for which the characteristic vectors of the minimal dominating sets are the zero-one solutions of a linear equality.
Formally, equidominating graphs have the following definition.

\begin{definition}
A graph $G=(V,E)$ is called \emph{equidominating} if there exists a value $t\in\IN = \{1,2,3,\ldots\}$ and a weight function $\vw \colon V \rightarrow \IN$ such that for all $D\subseteq V$ the following equivalence holds: 
\[D \text{ is an mds } \iff \vw(D) := \sum_{v\in D}{\vw(v)} = t.\] 
The pair $(\vw,t)$ is called an \emph{equidominating structure}, $\vw$ an \emph{equidominating function} and $t$ a \emph{target value}. 
\end{definition}

Figure~\ref{fig:example_equidom_graph} shows an equidominating graph. 
Every mds has a total weight of 23 and further, every subset of weight 23 is an mds.
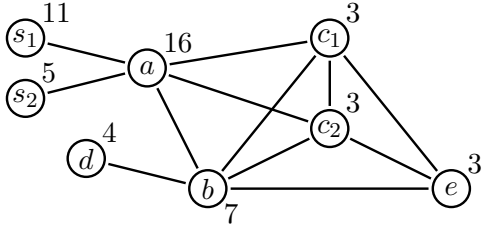
\begin{figure}	\label{fig:example_equidom_graph}
	\centering
	\psset{unit=0.4cm,radius=0.65,nodesep=0.1,linewidth=0.08,labelsep=0.1}
	\begin{pspicture}(-3,-1)(13,6)
			
		\Cnodeput(2,4){a}{$a$}
		\nput{45}{a}{16}
		\Cnodeput(4,0){b}{$b$}
		\nput{-45}{b}{7}
		\Cnodeput(8,5){c1}{$c_1$}
		\nput{45}{c1}{3}		
		\Cnodeput(8,2){c2}{$c_2$}
		\nput{45}{c2}{3}
		\Cnodeput(0,1){d}{$d$}
		\nput{45}{d}{4}
		\Cnodeput(12,0){e}{$e$}
		\nput{45}{e}{3}
		\Cnodeput(-2,5){s1}{$s_1$}
		\Cnodeput(-2,3){s2}{$s_2$}
		\nput{45}{s1}{11}
		\nput{45}{s2}{5}
		
		\ncline{a}{s1}
		\ncline{a}{s2}
		\ncline{a}{c1}
		\ncline{a}{c2}
		\ncline{a}{b}
		\ncline{c1}{c2}
		\ncline{c1}{e}
		\ncline{c2}{e}
		\ncline{c1}{b}
		\ncline{c2}{b}
		\ncline{b}{e}
		\ncline{b}{d}
		
	\end{pspicture}
	\caption{An equidominating graph on 8 vertices; the weights are drawn next to the vertices and the target value is $t=23$.}
\end{figure}
One advantage of having an equidominating structure of the graph at hand is that one can check whether a given vertex subset is an mds in linear time. 

The {\sc Equidomination} problem is to decide whether a given graph is equidominating or not.
Unfortunately, the computational complexity of this problem is unknown.
It is not even clear whether {\sc Equidomination} is in NP.

It can be seen that the following problem is coNP-complete:
given a graph $G$, a weight function $\vw$ and some number $t \in \mathbb N$, is $(\vw,t)$ an equidominating structure of $G$?
This intractability result remains true in the seemingly simple case when $G$ is just the disjoint union of edges.
That can be seen by applying literally the same reduction from the {\sc Weak Partition} problem given by Milani\v{c} et al.~\cite{MR2823204}, who proved coNP-completeness for the the analogous question for maximal stable sets.

We remark that there is no characterization of the class of equidominating graphs in terms of forbidden induced subgraphs:
if one attaches a pendant vertex to every vertex of an arbitrary graph (the so-called corona of a graph with $K_1$), the resulting graph is equidominating. 
The only existing result in this direction characterizes graphs that are both equidominating and domishold (see Theorem~2 in~\cite{Payan1980}).
Further, it is shown in~\cite{Payan1980} that threshold graphs are equidominating.

To get a grip on the computational complexity of the problem, we introduce the following two parameterized notions of equidomination.

\begin{definition}
For a given $t\in\IN$, a graph $G=(V,E)$ is called \emph{target-$\boldsymbol{t}$ equidominating} if there is an equidominating structure of the form $(\vw,t)$ for $G$.  
\end{definition}

\begin{definition}
For a given $k\in\IN$, a graph $G=(V,E)$ is called $\boldsymbol{k}$\emph{-equidominating} if there exists an equidominating structure $(\vw,t)$ with $\vw \colon V \rightarrow \{1,\ldots,k\}$ for some $t \in \IN$.
In this case, $\vw$ is said to be a $\boldsymbol{k}$\emph{-equidominating function} and the pair $(\vw,t)$ a $\boldsymbol{k}$\emph{-equidominating structure}.
\end{definition}

Note that a $k$-equidominating graph is also $k'$-equidominating for all $k'\geq k$.
It is clear that every target-$t$ equidominating graph is also $t$-equidominating since every vertex is contained in some mds and thus its weight, with respect to any equidominating function, cannot exceed~$t$.
The opposite, however, is not true.
Indeed, the edgeless graph on $t+1$ vertices is $k$-equidominating for every $k\in \IN$ but not target-$t$ equidominating.

In this paper, we study the following two parameterized versions of the {\sc Equidomination}~problem:

\vspace{1em}

\parbox{.9\linewidth}{
	{\sc $k$-Equidomination:}\\[0.5em]
	\begin{tabular}{rl}
		\textit{Instance:} & A graph $G$ and $k\in\IN$. \\
		\textit{Parameter:} & $k$. \\
		\textit{Problem:} & Decide whether $G$ is $k$-equidominating. \\
	\end{tabular}
}

\vspace{1em}

\parbox{.9\linewidth}{
	{\sc Target-$t$ Equidomination:}\\[0.5em]
	\begin{tabular}{rl}
		\textit{Instance:} & A graph $G$ and $t\in\IN$. \\
		\textit{Parameter:} & $t$. \\
		\textit{Problem:} & Decide whether $G$ is target-$t$ equidominating. \\
	\end{tabular}
}
\vspace{1em}

In the course of this paper we show that both the {\sc $k$-Equidomination} problem and the {\sc Target-$t$ Equidomination} problem are fixed-parameter tractable. 
We do this by using the so-called kernelization technique: 
we construct an equivalent instance the size of which is bounded by a function of the parameter.
Further, we give FPT algorithms for both problems by applying an XP algorithm to the kernels. 
For a given graph on $n$ vertices and $m$ edges, this leads to a running time of $\mathcal{O}\left(nm^2 + n^2 + t^{2t^2+3t+1}\right)$ for the {\sc Target-$t$ Equidomination} problem and $\mathcal{O}\left(nm^2 + n^2 + k^{6k^2+7k+1}\right)$ for the {\sc $k$-Equidomination} problem.

This paper is structured as follows: 
in Section~\ref{sec:preliminaries} we give a brief explanation of the needed fundamentals.
In Section~\ref{sec:propertiestwinclasses} we examine relationships between the twin relation and equidomination and give an algorithm to discover certain structures that can appear.
We develop an XP algorithm for the {\sc $k$-Equidomination} problem (which can also be used for the {\sc Target-$t$ Equidomination} problem) in Section~\ref{sec:xpalgorithm}.
We make use of this algorithm and of reduction rules described in Section~\ref{sec:reduction} to deduce the desired tractability results, stated in Section~\ref{sec:kernelkequidomination}.
Afterward, we characterize the class of hereditarily equidominating graphs in Section~\ref{sec:hereditarily}.
In Section~\ref{sec:conclusion} we draw a conclusion and give a brief outlook.

\section{Preliminaries} \label{sec:preliminaries}
In the first part of this section we give some basic definitions and notations that we use in this paper.
In the second part we give a brief introduction to fixed-parameter tractability and the third part is devoted to the so-called twin relation.

\subsection{Basic Definitions and Notations}
The graphs we work with all are undirected, simple and finite.
For a vertex $v\in V$ of a graph $G=(V,E)$ we call $N(v)\coloneqq \{x\in V\mid(x,v)\in E\}$ the \emph{open neighborhood} of $v$.
Further, $N[v]\coloneqq N(v)\cup\{v\}$ is the \emph{closed neighborhood} of $v$. 
We also use these terms for sets $S\subseteq V$, where $N(S)\coloneqq \bigcup_{v\in S}N(v)$ and $N[S]\coloneqq N(S)\cup S$.
We define the private neighbor set $pn[v,S]$ of $v\in S$ as $pn[v,S]\coloneqq N[v] \setminus N[S\setminus\{v\}]$ and every element of $pn[v,S]$ is called a \emph{private neighbor} of $v$.
Note that a set $D\subseteq V$ is dominating if and only if $N[D]=V$.
Further, $D$ is an mds if and only if additionally $pn[v,D]\neq\emptyset$ for all $v\in D$ holds. 
A \emph{stable set} is a subset of pairwise non-adjacent vertices and a \emph{clique} is a subset of pairwise adjacent vertices.
Analogously to an mds, a stable set is a \emph{maximal stable set} if it is not contained in another stable set.
It is easy to see that every maximal stable set is an mds.

The length of a shortest path between two vertices $v,w\in V$ is denoted by $dist(v,w)$ and called the \emph{distance} of $v$ and $w$.
For two vertices $x,y\in V$ and a subset $V'\subseteq V$ with $x\in V',\, y\notin V'$ we write shortly $V'-x+y$ for $(V'\setminus \{x\})\cup \{y\}$.
We use the abbreviation $[k]\coloneqq\{1,\ldots,k\}$ for a natural number $k\in\IN$.
For a subset $V'\subseteq V$ we denote the vertex induced subgraph of $G$ by $G[V']$.
Finally, we denote a disjoint union of two sets by $\disjcup$.

\subsection{Parameterized Complexity}
The concept of parameterized complexity is used to refine the analysis of computational hard problems.
It was formally introduced by Downey and Fellows in the 1990s \cite{DF1992}.
Parameterized complexity is motivated by the fact that often a large part of the running time of an algorithm depends not on the instance size but on a problem-specific parameter.
A parameterized decision problem is called \emph{fixed-parameter tractable} (FPT) if there exists an algorithm that solves it in time \hbox{$f(k) \cdot n^{O(1)}$}, where ~$n$ is the size of the instance and $k$ a problem-specific parameter. 
Such an algorithm is called an \emph{FPT algorithm}.
The function~$f$ is typically super-polynomial and it is essential that it depends only on~$k$.
A standard tool to prove that a parameterized problem is FPT and to design an FPT algorithm is the kernelization technique.
A \emph{kernelization} is a polynomial-time algorithm which transforms an instance~$(I,k)$ of a parameterized problem into an equivalent instance~$(I',k')$ of the same problem such that the size of~$I'$ is bounded by~$g(k)$ for some computable function~$g$ and $k'$ is bounded by a function of $k$.
The instance~$I'$ is said to be a \emph{kernel} of size~$g(k)$.
Furthermore, a transformation to a different but, loosely speaking, related problem, is called a \emph{generalized kernelization} and a transformed instance a \emph{generalized kernel} (see~\cite{BODLAENDER2009423}).
It is widely known that a parameterized problem is fixed-parameter tractable if and only if it admits a (generalized) kernelization.
In the remainder of this paper, the kernel size is expressed in terms of the number of vertices.
For more background on parameterized complexity the reader is referred to the books of Downey and Fellows~\cite{DF13} and Flum and Grohe~\cite{flumgrohe}.

\subsection{Twin Relation} \label{subsec:twinrelation}

Two vertices $v,\, w$ of a graph are called \emph{twins} if $N(v)\setminus\{w\}=N(w)\setminus\{v\}$ holds.
Here $v$ and $w$ can be either adjacent (\emph{true twins}) or non-adjacent (\emph{false twins}).
It is easy to see that the twin relation (two vertices are related if they are twins) is an equivalence relation.
The equivalence classes are called \emph{twin classes} and the partition of the vertices into twin classes is called the \emph{twin partition}.

All vertices of a twin class are either pairwise adjacent or pairwise non-adjacent.
Therefore, twin classes are specified to be \emph{clique classes} in the first and \emph{stable set classes} in the latter case.
A twin class can also be a single vertex. % which is then called a 
Even though a single vertex is strictly speaking a stable set as well as a clique, we use the terms clique class and stable set class only for twin classes with at least two elements. 
A twin class with one vertex is called a \emph{singleton class}.
In Figure~\ref{fig:example_equidom_graph}, $c_1$ and $c_2$ form a clique class, $s_1$ and $s_2$ a stable set class, and all other vertices singleton classes.

Let $T_1$ and $T_2$ be two twin classes.
Then either every vertex of $T_1$ is adjacent to every vertex of $T_2$ or every vertex of $T_1$ is non-adjacent to every vertex of $T_2$.
In the first case we say that $T_1$ and $T_2$ \emph{see} each other and that $T_1$ sees $T_2$ and vice versa.
We also say that a vertex and a twin class see each other, and likewise two vertices.
Furthermore, if appropriate we use expressions for twin classes which are usually used for vertices (e.g.~a twin class is adjacent to, dominates, is dominated by,~...).

For the sake of completeness we define the \emph{quotient graph} $Q(G)$ of a graph $G$, where every twin class becomes a vertex and two vertices are adjacent if and only if the corresponding twin classes see each other.
The twin partition and the quotient graph can be computed in linear time using one of the modular decomposition algorithms of \cite{ModDecomp1,ModDecomp2,ModDecomp3}.

\section{Properties of Twin Classes regarding Equidomination} \label{sec:propertiestwinclasses}
In this section we use the twin relation to obtain structural results regarding equidomination.
Initially, we introduce the following definition.

\begin{definition}
Let $G=(V,E)$ be a graph.
Two vertices $x,y\in V$ are called \emph{mds-exchangeable} if and only if there exists an mds $D\subseteq V$ with $|\{x,y\}\cap D|=1$ and if for all mds $D\subseteq V$ with $|\{x,y\}\cap D|=1$ the symmetric difference $(D \setminus\{x,y\}) \cup (\{x,y\} \setminus D)$ is also an mds.
\end{definition}
Loosely speaking, two vertices are mds-exchangeable if they can be exchanged for another in any mds containing exactly one of them.
Two mds-exchangeable vertices of an equidominating graph must have the same weight in every equidominating function.
Even if two vertices are mds-exchangeable, they can both be elements of one mds (e.g. two non-adjacent vertices of $C_4$).
While $c_1$, $c_2$ and $e$ in Figure~\ref{fig:example_equidom_graph} are mds-exchangeable, $s_1$ and $s_2$ are not since there is no mds containing only one of the two vertices.

As the following observations show, the twin relation is a very helpful instrument with regard to equidomination.  

\begin{observation}
\label{obs:keinereineroderalleinSSC}
	For every mds $D$ and every stable set class $S$, we have $|D\cap S|\in\{0,1,|S|\}$.
\end{observation}

\begin{observation}
\label{obs:einknotenproCC}
	For every mds $D$ and every clique class $C$, we have $|D\cap C|\leq 1$.
\end{observation}

\begin{observation}
\label{obs:constantaufklassen}
Since the vertices of a clique class are pairwise mds-exchangeable, every equidominating function must be constant on every clique class.
Analogously, if there exists an mds containing exactly one vertex of a stable set class, every equidominating function must be constant on that stable set class.
\end{observation}

\begin{observation}\label{obs: SSC-maximal-t-knoten}
Since every maximal stable set is also an mds, every (non-maximal) stable set can be extended to an mds.
Therefore, for every stable set $S$ and for every equidominating structure $(\vw,t)$ it holds that $\vw(S) \leq t$ and hence $|S|\leq t$.
\end{observation}

Such an extension of a stable set to an mds can happen in different ways.
One way is the following: for a given stable set $S\subseteq V$ of a graph $G=(V,E)$ we determine an mds $D'$ of $G[V\setminus N[S]]$.
Then, the union $D\coloneqq S\cup D'$ is an mds of $G$.
The domination property of $D$ is obvious. 
Furthermore, every vertex of $S$ is its own private neighbor (possible among others) and for every vertex $v\in D'$ it holds that $pn(v,D')=pn(v,D)$.
That means it is not possible to delete one vertex and remain dominating.
Hence, for every stable set $S$ there exists an mds $D$ with $S\subseteq D$ and $D\cap N(S)=\emptyset$, which we will often use in the proofs of the upcoming lemmas.

In the following lemmas we examine whether different vertices of an equidominating graph can have equal weights or not.
As we will see, two vertices can only have the same weight if they lie in the same twin class or are adjacent.
That means for one thing, that when trying to construct an equidominating structure one has to consider fewer combinatorial possibilities.
And for another thing, that for a given number of weights to be allocated one can bound the diameter of an equidominating graph.

\begin{lemma}
\label{lem:verschiedenegewichteabstandzwei}
Let $G=(V,E)$ be an equidominating graph with equidominating structure $(\vw,t)$ and let $x,\,y \in V$ be two vertices of different twin classes with $dist(x,y)\geq 2$.
Then it holds that $\vw(x)\neq\vw(y)$.
\end{lemma}
\begin{proof}
Suppose there are such two vertices $x, y\in V$ with $\vw(x)=\vw(y)$.
First, let one of the vertices be in a clique class $C$, say $x\in C$.
Extend $\{x,y\}$ to an mds and then exchange another vertex of $C$ for $y$.
Doing so, one gets a subset with weight $t$, which is not an mds (cf.~Observation~\ref{obs:einknotenproCC}), a contradiction. 

So, let both $x$ and $y$ be either elements of singleton classes or stable set classes.
As they are in different classes, w.l.o.g.\ let $v$ be a vertex only seen by $x$.
Extend $\{y,v\}$ to an mds $D$ such that $D\cap N(\{y,v\})=\emptyset$ and exchange $y$ for $x$.
Again we get a subset (namely $D-y+x$) of weight $t$, that is not an mds since $y$ is no longer dominated. 
\end{proof}

Note that in the previous lemma the two mentioned vertices must be from different twin classes.
Two elements of a stable set class of course can have the same weight while always having distance at least two. 

In the following we take a closer look at adjacent vertices, where we find a slightly more complicated situation.
We begin by showing that vertices of stable set classes and clique classes that see each other cannot have the same weight.

\begin{lemma} 
Let $G=(V,E)$ be an equidominating graph with equidominating structure $(\vw,t)$ and let $S\subseteq V$ be a stable set class and $C\subseteq V$ a clique class that see each other.
Then for all $x\in S$ and for all $y\in C$ it holds that $\vw(x)\neq\vw(y)$.
\end{lemma}
\begin{proof}
Suppose there are vertices $x\in S$ and $y\in C$ with $\vw(x)=\vw(y)$.
The other vertices of $S$ must have a different weight than $x$, since otherwise one could extend $S$ to an mds $D$ such that $D\cap N(S)=\emptyset$ and exchange two vertices of $S$ for two vertices of $C$ while maintaining the weight $t$, a contradiction.
So, extend $S$ to an mds $D$ such that $D\cap N(S)=\emptyset$.
The set $D'\coloneqq D-x+y$ has weight $t$ and, therefore, is also an mds.
Now $D'-x'+x$ ($x'\in S,\, x'\neq x$) has a weight different than $t$, but since the swapped vertices lie in one twin class it still must be an mds, again a contradiction.
\end{proof}

Further, a vertex of stable set class and an adjacent singleton class cannot have the same weight in an equidominating structure.

\begin{lemma} \label{lem:SSCundsingleton}
	Let $G=(V,E)$ be an equidominating graph with equidominating structure $(\vw,t)$.
	Let $S\subseteq V$ be a stable set class and let $\{y\}$ be a singleton class with $y\in N(S)$.
	Then for all $x\in S$ it holds that $\vw(x)\neq\vw(y)$.
\end{lemma}

\begin{proof}
	Let $x,x'\in S$, $y\in N(S)$ an adjacent singleton class and suppose that $\vw(x)=\vw(y)$.
	The set $N(S)\setminus N[y]$ cannot be empty, since otherwise one could extend $\{y\}$ to an mds $D$ such that $D\cap N(y)=\emptyset$, but $D-y+x$ does not dominate $x'$, which contradicts $\vw(x)=\vw(y)$.
	Furthermore, there must be a vertex $v\in N(S)\setminus N[y]$ with $N[v]\subseteq N[S]$ (otherwise $\{y\}\cup (V\setminus N[S])$ is a dominating set, which contains an mds $D$ with $y\in D$, but $D-y+x$ does not dominate $x'$).

	So, extend $\{v,y\}$ to an mds $D$ such that $D\cap N(\{v,y\})=\emptyset$.
	Since $D'\coloneqq D-y+x$ has weight $t$ and thus is an mds, $pn(v,D')=S\setminus \{x\}$.
	It follows that $D'-v+y$ is also an mds, which implies $\vw(v)=\vw(y)$, a contradiction to Lemma~\ref{lem:verschiedenegewichteabstandzwei}.
\end{proof}

The proof of Lemma~\ref{lem:SSCundsingleton} suggests that if $v$ and $y$ lie in one twin class, $x$ and $y$ could have the same weight.
As the next Lemma shows this is indeed possible, but only in a specific situation.

\begin{lemma} \label{lem:SSCmitgleichemgewicht}
Let $G=(V,E)$ be an equidominating graph with equidominating structure $(\vw,t)$ and let $S_1,S_2\subseteq V$ be two adjacent stable set classes with two vertices $x\in S_1,\, y \in S_2$ of the same weight.
Then the following assertions hold:
\begin{enumerate}[label={(\roman*)}]
	\item $|S_1|=|S_2|=2$, \label{lem:SSCmitgleichemgewicht_2elements}
	\item $\vw$ is constant on $S_1\cup S_2$, \label{lem:SSCmitgleichemgewicht_const}
	\item\label{lem:SSCmitgleichemgewicht_sameneighbors} every twin class seen by $S_1$ is also seen by $S_2$ and vice versa. 
\end{enumerate}
Further, if two adjacent stable set classes of size two have the same closed neighborhood, all vertices of those stable set classes have the same weight in any equidominating structure.
\end{lemma}
\begin{proof}
Suppose one of the stable set classes has more than two elements, say $S_1$. We extend $S_1$ to an mds $D$ with $D\cap S_2=\emptyset$. 
Then, $D-x+y$ hat weight $t$ but following Observation~\ref{obs:keinereineroderalleinSSC} it is not an mds.
So, assertion~\ref{lem:SSCmitgleichemgewicht_2elements} is shown and let $S_1=\{x,x'\}$ and $S_2=\{y,y'\}$.

Let $D^*$ be an mds of $G[V\setminus N[S_1\cup S_2]]$. Then, $D\coloneqq D^*\cup\{x,y\}$ is an mds of G (note that $y'\in pn(x,D)$ and  $x'\in pn(y,D)$).
Observation~\ref{obs:constantaufklassen} yields assertion~\ref{lem:SSCmitgleichemgewicht_const}. 

Suppose $S_2$ sees a vertex $v$ which is not seen by $S_1$.
Extend $S_1\cup\{v\}$ to an mds $D^*$ such that $D^*\cap S_2=\emptyset$.
Of course $D\coloneqq(D^*\setminus S_1) \cup S_2$ has weight $t$ but $D$ is not an mds, as $N[D]=N[D\setminus\{y\}]$.
That proves assertion~\ref{lem:SSCmitgleichemgewicht_sameneighbors}.

The last statement of the lemma follows from the fact, that the vertices of the two adjacent stable set classes of size two with identical closed neighborhood are pairwise mds-exchangeable. 
\end{proof}

As a consequence of Lemma~\ref{lem:SSCmitgleichemgewicht} there can be an arbitrarily large number of stable set classes of size two with vertices of the same weight.
Such an occurrence could be a problem when trying to achieve bounded kernels for the parameterized problems.
But the good thing is that all of those stable set classes both see each other and see the same twin classes in the remainder of the graph.
Therefore, as we will see, it is possible to reduce them to a manageable number.
For a better handling we introduce the following new term. 
For that we define $K_{2n}-ne$ ($n\in\IN$) to be the graph obtained from a complete graph on $2n$ vertices by removing $n$ disjoint edges.
\begin{definition}
	Let $G=(V,E)$ be a graph and $\mathcal{S}\subseteq V$ be a maximal subset such that $G[\mathcal{S}] \cong K_{2n}-ne$ for some $n\geq 2$ and such that each vertex of $\mathcal{S}$ is adjacent to the same vertices in $V\setminus\mathcal{S}$.
	Then $\mathcal{S}$ is called a \emph{stable set bundle}.
\end{definition}
Maximal here means that no other subset  fulfills the two conditions and properly contains $\mathcal{S}$.
Every stable set bundle contains several stable set classes of size two.
Note that a stable set bundle forms a clique class in the quotient graph.
Considered the other way around, adding a false twin to every vertex of a clique class yields a stable set bundle.
Following Lemma~\ref{lem:SSCmitgleichemgewicht} the vertices of a stable set bundle are pairwise mds-exchangeable and, therefore, every equidominating function is constant on a stable set bundle.

Now, regarding the question whether two vertices can have the same weight in an equidominating structure, the last open question is:
can vertices of a clique class or a singleton class have the same weight as its neighboring clique class or singleton class?
The answer is yes: there can be a clique, that is not a clique class, whose vertices are pairwise mds-exchangeable.
\begin{definition}
	Let $G$ be a graph and $\mathcal{C}$ an inclusion-wise maximal clique of pairwise mds-exchangeable vertices that contains at least two twin classes.
	Then $\mathcal{C}$ is called a \emph{clique bundle}.
\end{definition}

Upon first reading it seems a little bit odd to define clique bundles exactly as what we are looking for: 
pairwise mds-exchangeable vertices of possibly different twin classes.
However, the crucial thing here is that we can identify clique bundles efficiently (see Algorithm~\ref{alg:adjacentmdsexchangeable}).
We require at least two twin classes to be in a clique bundle in order that a twin class on its own is not a clique bundle and a twin class at the same time.

In a clique bundle there can be both clique classes and singleton classes but no stable set classes.
In the graph shown in Figure~\ref{fig:example_equidom_graph} the clique class $\{c_1,c_2\}$ and the singleton class $\{e\}$ form a clique bundle.
In Figure~\ref{fig:cliquebundle} you can see an equidominating graph which consists of two clique bundles.
In this graph every mds contains exactly one vertex of each clique bundle.

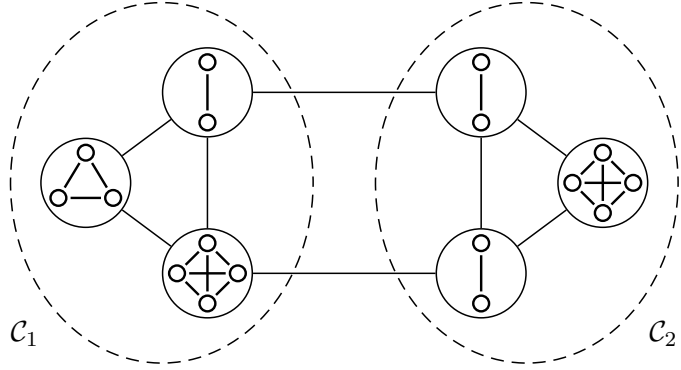
\begin{figure} 
	\centering
	\psset{unit=0.4cm,radius=0.3,nodesep=0.1,linewidth=0.08}
	\begin{pspicture}(0,0)(21,12)
		\Cnode(1.1,5.5){a11}
		\Cnode(2.9,5.5){a12}
		\Cnode(2,7){a13}
		\ncline{a11}{a12} \ncline{a11}{a13} \ncline{a13}{a12}
		
		\Cnode(6,10){a21}
		\Cnode(6,8){a22}
		\ncline{a21}{a22}
		
		\Cnode(6,4){a31}
		\Cnode(6,2){a32}
		\Cnode(5,3){a33}
		\Cnode(7,3){a34}
		\ncline{a31}{a32} \ncline{a31}{a33} \ncline{a31}{a34} \ncline{a32}{a33} \ncline{a32}{a34} \ncline{a33}{a34}
		
		\Cnode(15,10){b11}
		\Cnode(15,8){b12}
		\ncline{b11}{b12}

		\Cnode(15,4){b21}
		\Cnode(15,2){b22}
		\ncline{b21}{b22}
		
		\Cnode(19,7){b31}
		\Cnode(19,5){b32}
		\Cnode(18,6){b33}
		\Cnode(20,6){b34}
		\ncline{b31}{b32} \ncline{b31}{b33} \ncline{b31}{b34} \ncline{b32}{b33} \ncline{b32}{b34} \ncline{b33}{b34}
	
		\psset{nodesep=0,linewidth=0.05}
		
		\Cnode[radius=1.5](2,6){a1}
		\Cnode[radius=1.5](6,9){a2}
		\Cnode[radius=1.5](6,3){a3}
		\Cnode[radius=1.5](15,9){b1}
		\Cnode[radius=1.5](15,3){b2}
		\Cnode[radius=1.5](19,6){b3}			
		
		\ncline{a1}{a2}
		\ncline{a1}{a3}
		\ncline{a2}{a3}

		\ncline{b1}{b2}
		\ncline{b1}{b3}
		\ncline{b2}{b3}

		\ncline{a2}{b1}
		\ncline{a3}{b2}
		
		\psset{linestyle=dashed}
		
		\psellipse(4.5,6)(5,6)
		\psellipse(16.5,6)(5,6)
	
		\rput(0,1){$\mathcal{C}_1$}
		\rput(21,1){$\mathcal{C}_2$}
				
	\end{pspicture}
	\caption{Example of an equidominating graph consisting of the two clique bundles $\mathcal{C}_1$ and $\mathcal{C}_2$ each containing three clique-classes with two to four vertices. The clique classes are indicated by circles and an edge between the circles of two clique classes represents all edges between the vertices of the corresponding clique classes. \label{fig:cliquebundle} }

\end{figure}

We use the term \emph{bundle} to refer to either a stable set bundle or a clique bundle.
Recall that a twin class is either a stable set class, a clique class or a singleton class.
Due to the existence of bundles we introduce a sort of generalization of twin classes:
a \emph{pseudo class} is either a twin class not contained in a bundle or a stable set bundle or a clique bundle.
That is, a pseudo class is exactly one of following: a) a singleton class, b) a stable set class, c) a clique class, d) a stable set bundle or e) a clique bundle.
By this definition and the previous discussion we get the following result.

\begin{corollary}\label{cor:unique_pseudo_class_partition}
	There is a unique partition of the vertices of a graph into pseudo classes.
\end{corollary}
The introduction of pseudo classes is motivated by the following corollary, which is a summary of this section.

\begin{corollary} \label{cor:differentpseudodifferentweights}
Let $G$ be an equidominating graph with equidominating structure $(\vw,t)$ and $P_1,P_2$ be two different pseudo classes. Then for all $x\in P_1$ and $y\in P_2$ it holds that $\vw(x)\neq\vw(y)$.
\end{corollary}

To recognize bundles we developed Algorithm~\ref{alg:adjacentmdsexchangeable}, which decides for two adjacent vertices whether or not they are mds-exchangeable.
More precisely, the algorithm checks if there can be private neighbors of one vertex in any dominating set that are not seen by the other one.
After computing the twin partition one can apply Algorithm~\ref{alg:adjacentmdsexchangeable} to adjacent clique classes and singleton classes, and to adjacent stable set classes of size two to find all clique bundles and stable set bundles, respectively.

\begin{algorithm} 
	\caption{Checking adjacent vertices for mds-exchangeability} 
	\label{alg:adjacentmdsexchangeable}
	\begin{algorithmic}[1]
		\Require{Two adjacent vertices $x,y\in V$ of a graph $G=(V,E)$}
		\Ensure{\textbf{YES}, if $x$ and $y$ are mds-exchangeable, otherwise \textbf{NO}}
		\ForAll{$(v_1,v_2)\in\{(x,y),(y,x)\}$} \Comment{Check both combinations}
			\ForAll{$v' \in N(v_1) \setminus N[v_2]$} %\Comment{($\ast$)}
				\If {$\{v_1\}\cup \Big( V(G)\setminus \big(N[v']\cup\{v_2\}\big)\Big)$ is a dominating set} 
					\State return \textbf{NO} \Comment{$x$ and $y$ are not mds-exchangeable}
				\EndIf
			\EndFor
		\EndFor
		\State return \textbf{YES}  \Comment{$x$ and $y$ are mds-exchangeable}
	\end{algorithmic}
\end{algorithm}

\begin{proposition}
Algorithm~\ref{alg:adjacentmdsexchangeable} is correct and runs in $\mathcal O(nm)$ time.
\end{proposition}
\begin{proof}
First we prove the correctness.
Since the algorithm checks both combinations (cf.\ line~1), the following argumentation holds for $v_1=x,\,v_2=y$ as well as for $v_1=y,\,v_2=x$. 
If the output is NO, the stated set in line~3 contains an mds $D$ with $v'\in pn(v_1,D)$ and $v'\notin N[v_2]$.
So, $v'$ is not dominated by $D-v_1+v_2$ and thus $x$ and $y$ cannot be mds-exchangeable.
If the output is YES and the algorithm reaches line 8, there is no vertex $v'$ in $N(v_1)$, which is not adjacent (or equal) to $v_2$ and which is a private neighbor of $v_1$  in any mds that does not contain $v_2$. %(possibly $v_1$ its own private neighbor) <-??
So, we cannot find an mds $D$ with $v_1\in D$,  $v_2\notin D$ and $pn(v_1,D)\cap N[v_2]\neq\emptyset$.
Further, since $v_1$ and $v_2$ are adjacent, there are mds that contain only $v_1$ and that contain only $v_2$.
Therefore, $x$ and $y$ are mds-exchangeable.

It is clear that the algorithm terminates for a finite graph.
There are $\mathcal O(n)$ vertices in $N(v_1) \setminus N[v_2]$.
For each vertex $v' \in N(v_1) \setminus N[v_2]$ it is sufficient to check if the vertices of $N(v')\setminus\{v_1\}$ are dominated since every other vertex is an element of the set stated in line~3.
With the use of two flags for each vertex and a global counter we can do this considering every edge two times. 
So, line~3 needs $\mathcal O(m)$ time which leads to a total running time of $\mathcal O(nm)$.
\end{proof}

To discover all bundles one has to apply Algorithm~\ref{alg:adjacentmdsexchangeable} for every edge, which gives us a total running time of $\mathcal O(nm^2)$ to determine the pseudo class partition.
	
\section{An XP algorithm for the \textit{k}-{\sc Equidomination} Problem} \label{sec:xpalgorithm}
In this section we describe an algorithm which decides whether a given graph is $k$-equidominating for some fixed $k\in \IN$ with a running with only $k$ appearing in the exponents.
The aim is to apply this algorithm to the constructed kernels of the parameterized problems.
The algorithm mainly follows the ideas and the algorithm for the $k$-{\sc Equistability} problem of Levit et al.~\cite{kequistable_fpt,Recog_kequistable}. 
However, it has to be extended due to the existence of clique bundles and stable set bundles. 

The basic idea of the algorithm is that by considering the pseudo classes one does not have to examine every possible weight function nor every possible subset of vertices.
Since different vertices of the same pseudo class, roughly said, play the same role regarding domination, two weight functions that differ only by switched weights for two vertices of one pseudo class, can be handled as the same.
This leads to equivalence classes of weight functions.
Further, we reduce the running time from a brute force algorithm by classifying subsets of vertices.
It does not matter, for example, which vertex of a clique class is in an mds.
Then, we have to check only one subset per class for being an mds.

\begin{theorem}
For a given $k\in\IN$ it is decidable if a graph $G=(V,E)$ with $|V|=n$ and $|E|=m$ is $k$-equidominating or not in time $\mathcal{O}\left(nm^2 + n^kk^k + n^{2k+2}k^{-k-1} + k^{3k+3}\right)$ and a $k$-equidominating structure is computed in this time.
Further, for $t\in\IN$ the same algorithm can be used to decide if $G$ is target-$t$ equidominating.
\end{theorem}
\begin{proof}
We first discuss how the algorithm works and calculate its running time afterward.
There are $k^n$ different weight functions from $V$ to $[k]$ we would have to test and $2^n$ potential mds.
As we will see, we can reduce both numbers using the pseudo classes. 

\begin{algorithm}
	\caption{An XP algorithm for the \textit{k}-{\sc Equidomination} problem} 
	\label{alg:xp_algorithm}
	\begin{algorithmic}[1]
		\Require{A graph $G=(V,E)$, $k\in\IN$}
		\Ensure{a $k$-equidominating structure if $G$ is $k$-equidominating, otherwise \textbf{NO}}
		\State determine pseudo classes $P_1,\ldots P_s$ by computing the twin partition and identifying clique bundles and stable set bundles via Algorithm~\ref{alg:adjacentmdsexchangeable} \label{alg:xp_algorithm-line:pseudoclasses}
		\If {$s>k$} 
			\State return \textbf{NO}	\Comment{$G$ cannot be $k$-equidominating}
		\EndIf
		\State compute an arbitrary mds $D$
		\State compute the set $\Omega$ of weight functions to check \label{alg:xp_algorithm-line:compute_omega}
		\ForAll {$\vw\in\Omega$}
			\State compute $t_\vw=\sum_{v\in D}\vw(v)$ \label{alg:xp_algorithm-line:compute_t}
			\ForAll {$x\in X_\vw$}
				\State compute an arbitrary $S\in\mathcal{S}_\vw(x)$ \label{alg:xp_algorithm-line:compute_S}
				\If {$S$ is an mds}
					\If {$\sum_{i=1}^k ix_i\neq t_\vw$} \label{alg:xp_algorithm-line:sum1}
						\State \textbf{next} $\vw$	\Comment{there is an mds of weight unequal to $t$}
					\EndIf
				\ElsIf {$\sum_{i=1}^k ix_i= t_\vw$}	\label{alg:xp_algorithm-line:sum2}
						\State \textbf{next} $\vw$	\Comment{there is a set of weight $t$ which is not an mds}
				\EndIf
				\State return $(\vw,t_\vw)$	\Comment{$\vw$ is a $k$-equidominating function}
			\EndFor
		\EndFor
		\State return \textbf{NO}	\Comment{no $k$-equidominating function was found}
	\end{algorithmic}
\end{algorithm}

Let $(P_1,\ldots,P_s)$ be the partition of $V$ into pseudo classes (note that $s\leq k$, otherwise $G$ is not $k$-equidominating).
We define an equivalence relation on the set of weight functions as follows: 
\[\vw_1 \sim \vw_2 :\Longleftrightarrow \forall j \in [s] ~ \forall i \in [k]\colon  |\vw_1^{-1}(i)\cap P_j|=|\vw_2^{-1}(i)\cap P_j|\,.\]
In simple terms two weight functions are equivalent if they assign each weight to the same number of vertices within every pseudo class.
It is easy to verify that either all functions of an equivalence class are $k$-equidominating functions or none.
That means we only have to check one representative per equivalence class.

Following Corollary~\ref{cor:differentpseudodifferentweights} it is a necessary condition for every equidominating function that elements of different pseudo classes have different weights.
Let $(v_1,\ldots,v_n)$ be an fixed ordering of $V$ in which vertices of the same pseudo class appear sequentially.
We define $\Omega$ to be the set of all weight functions $\vw \colon V \rightarrow [k]$ for which the weights also appear consecutively in the weight vector $(\vw(v_1),\ldots,\vw(v_n))\in [k]^n$ and which do not allocate same weights to different pseudo classes.
Now it is sufficient to check every function of $\Omega$ to decide whether $G$ is $k$-equidominating since every $k$-equidominating function not in $\Omega$ has an equivalent weight function in $\Omega$. 

In a second step we have to check for each $\vw\in\Omega$ if it is a $k$-equidominating function.
For that, let $\vw\in\Omega$, $D\subseteq V$ be an arbitrary mds and $t_\vw\coloneqq\vw(D)$.
We define $X_\vw\coloneqq\{x\in\IZ^k\mid 0\leq x_i\leq \vert \vw^{-1}(i)\vert\}$. % and use a analogous approach as above to reduce the number of subsets to work with:
Every vector $x\in X_\vw$ encodes a set of subsets $\mathcal{S}_\vw(x)\subseteq \mathcal{P}(V)$ (with $\mathcal{P}(V)$ being the power set of $V$) in the following sense:
The value $x_i$ of the $i$-th coordinate of $x$ is equal to the vertices of weight $i$ being in a subset, i.e.\ $\mathcal{S}_\vw(x)=\{S\subseteq V\mid \vert S\cap\vw^{-1}(i)\vert= x_i,\, i=1,\ldots,k\}$.
Since only vertices within a pseudo class have equal weights and since by construction it is regardless for being an mds which vertices of a pseudo class are in a subset of $V$, either all elements of $\mathcal{S}_\vw(x)$ are mds or none.
(In fact, implicitly we have defined an equivalent relation on $\mathcal{P}(V)$ and the $\mathcal{S}_\vw(x)$ are the equivalence classes.)
The last thing to do is to check for each $x\in X_w$ whether $x$ encodes mds if and only if $\sum_{i=1}^k ix_i=t_\vw$. In conclusion we achieve Algorithm~\ref{alg:xp_algorithm}.

We now calculate the running time of the above described algorithm.
The determination of the pseudo classes in line~\ref{alg:xp_algorithm-line:pseudoclasses} needs $\mathcal O(n+m+nm^2)$ time.
An arbitrary mds $D$ can be calculated straightforward in time $\mathcal O(n+m)$. 
A weight function can be described by its weight vector regarding the fixed ordering.
A basic, combinatorial result says that there are $\binom{n+k-1}{n}$ possibilities to divide the fixed ordered vertices $(v_1,\ldots,v_n)$ into $k$ (possibly empty) intervals, whereat an interval only contains consequently in the ordering appearing vertices.
There are $k!$ ways to distribute $k$ weights one-to-one to the intervals, which leads to $|\Omega|\leq k!\binom{n+k-1}{n}$.
With $\hat{n}\coloneqq\max\{n,k^2\}$ we can further estimate the number of potential $k$-equidominating functions to be $|\Omega|=\mathcal O(\hat{n}^k/k)$ and $\Omega$ can be computed in time $\mathcal O((k\hat{n})^k)$ (for more details see \cite{kequistable_fpt}). The lines~\ref{alg:xp_algorithm-line:compute_t} and \ref{alg:xp_algorithm-line:compute_S} can be executed in time $\mathcal O(n)$.
The numbers of vectors in $X_\vw$ is $\Pi_{i=1}^k(|\vw^{-1}(i)|+1)$, which is bounded by $(n/k+1)^k=\mathcal O((\hat{n}/k)^k)$ (see Lemma~4 in \cite{kequistable_fpt}).
The sums in line~\ref{alg:xp_algorithm-line:sum1} and \ref{alg:xp_algorithm-line:sum2} are calculated in time $\mathcal O(k)$.
Finally, it takes time $\mathcal O(n^2)$ to check if a subset of $V$ is an mds. 

Taken together we obtain a total running time of Algorithm~\ref{alg:xp_algorithm} of
\begin{align*}
	 & \mathcal{O}\left( n + m + nm^2 + (\hat{n}k)^k + \frac{\hat{n}^k}{k}\left(n+ \left(\frac{\hat{n}}{k}\right)^k\left(n+n^2+k\right) \right) \right) \\
	 & \hspace{2em} = \mathcal{O}\left(nm^2 + \left(\hat{n}k\right)^k + \frac{\hat{n}^{2k}n^2}{k^{k+1}}+\frac{\hat{n}^{2k}}{k^{k}}\right) \\
	 & \hspace{2em} = \mathcal{O}\left(nm^2 + n^kk^k + \frac{n^{2k+2}}{k^{k+1}}+\frac{n^{2k}k}{k^{k+1}} +
																	k^{3k} + \frac{k^{4k}n^2}{k^{k+1}}+\frac{k^{4k}}{k^{k}}\right) \\
	 & \hspace{2em} = \mathcal{O}\left(nm^2 + n^kk^k + n^{2k+2}k^{-k-1} + k^{3k+3}\right).																	
\end{align*}

	To check if a graph is target-$t$ equidominating, we simply set $k\coloneqq t$ and instead of defining $t_\vw\coloneqq\vw(D)$ in line~\ref{alg:xp_algorithm-line:compute_t}, we proceed only with those weight functions $\vw\in\Omega$ for which $\vw(D)=t$ holds.
\end{proof}

Basically, the algorithm considers all subsets of vertices.
Therefore it can be easily modified to provide the domination number and the upper domination number as well as mds of minimum and maximum cardinality.

\section{Reduction Rules} \label{sec:reduction}

In the following we examine three reduction rules which we use to construct (generalized) kernels of the \textsc{Target}-$t$ \textsc{Equidomination} problem as well as the $k$-\textsc{Equidomination} problem.
A graph is called \emph{ target-$\boldsymbol{t}$ $\boldsymbol{k}$-equidominating} if there is a $k$-equidominating structure with target value $t$.
Note that this is stronger than being $k$-equidominating and target-$t$ equidominating. 

\begin{lemma}\label{lem:reductions_general}
	Let $G$ be a graph, $r,\,k\in \IN$ and $M\subseteq V(G)$ a subset of pairwise mds-exchangeable vertices with $|M|>r$.
	Furthermore, let $G'$ be the graph obtained from $G$ by deleting all but $r$ vertices of $M$.
	If  
	\begin{enumerate}[label={(\roman*)},ref={\thetheorem(\roman*)}]
		\item\label{lem:reductions_general:Vor:hochstens-einer}  
			$|D\cap M|\leq r$ for every mds $D\subseteq V(G)$ of $G$,
		\item\label{lem:reductions_general:Vor:austauschbar}
			the vertices of $M\cap V(G'$) are pairwise mds-exchangeable, and
		\item\label{lem:reductions_general:Vor:dom-set} 
			every dominating set $D\subseteq V(G')$ of $G'$ is a dominating set of $G$,
	\end{enumerate}
	then the following equivalence holds for all $t\leq r$:
	\[G \text{ is target-}t \ k\text{-equidominating} \Longleftrightarrow G' \text{ is target-}t \ k\text{-equidominating}.\]
\end{lemma}
\begin{proof}
	To start, we claim that a subset $D\subseteq V(G')$ is an mds of $G'$ if and only if $D$ is an mds of $G$.
	For that, let $D\subseteq V(G')$.
	
	First, let $D$ be an mds of $G'$.
	By assumption~\hyperref[lem:reductions_general:Vor:dom-set]{(iii)}, $D$ is a dominating set of $G$.
	Suppose that $D$ is not minimal in $G$.
	Then there is a proper subset $\widetilde{D}\subsetneq D$ that is a dominating set of $G$.
	Since $G'$ is an induced subgraph of $G$, $\widetilde{D}$ is a dominating set of $G'$, a contradiction.
	
	Secondly, let $D$ be an mds of $G$. 
	Again it follows that $D$ is a dominating set of $G'$.
	Suppose $D$ is not minimal in $G'$.
	Then there is a proper subset $\widetilde{D}\subsetneq D$ that is an mds of $G'$.
	It follows from the preceding paragraph that $\widetilde{D}$ is also an mds of $G$.
	This contradicts the minimality of $D$ in $G$ and the claim is proved.	
	
	Now, we prove the equivalence stated in the lemma.
	Let $t\leq r$ and $M'\coloneqq M\cap V(G')$.

	\uline{$\Longrightarrow$:} Let $G$ be target-$t$ $k$-equidominating and $(\vw,t)$ be a $k$-equidominating structure of $G$.
	Any subset of $V(G')$ is an mds of $G'$ if and only if it is an mds of $G$, which in turn is the case if and only if it has total weight $t$.	
	It follows that $(\vw',t)$ with $\vw'\coloneqq\vw\big|_{V(G')}$ is a $k$-equidominating structure of $G'$ with target value $t$.

	\uline{$\Longleftarrow$:} Let $G'$ be target-$t$ $k$-equidominating and $(\vw',t)$ be a $k$-equidominating structure of $G'$.
	We define 
	\[\vw(v)\coloneqq \begin{cases} \vw'(v), & \text{if }v \in  V(G'),\\
																	\vw'(w), & \text{otherwise}, \end{cases}\]
	for any $w \in  M\cap V(G')$. 
	Since every equidominating function is constant on a set of pairwise mds-exchangeable vertices, the choice of $w$ is irrelevant for the definition of $\vw$.
	We show that $(\vw,t)$ is a $k$-equidominating structure of $G$.
	It is clear that $\vw(v)\leq k$ for every $v\in V(G)$.

	Let $X\subseteq V(G)$ be a subset of vertices of $G$ with $\vw(X)=t$. 
	As $\vw(X\cap M)\leq\vw(X)=t$, we get $|X\cap M|\leq t$.
	Therefore, we can exchange vertices of $X\setminus V(G')$, if any, with vertices of $M'\setminus X$ (one to one) to obtain a subset $\widetilde{X}\subseteq V(G')$ with $\vw'(\widetilde{X})=\vw(\widetilde{X})=t$.
	It follows that $\widetilde{X}$ is an mds of $G'$ and thus of $G$, too.
	To construct $\widetilde{X}$, we only exchanged vertices with each other that are mds-exchangeable.
	It follows that $X$ is an mds of $G$.

	Now, let $D\subseteq V(G)$ be an mds of $G$.
	By assumption~\hyperref[lem:reductions_general:Vor:hochstens-einer]{(i)}, we have $|D\cap M|\leq r$.
	As before, we can exchange vertices of $D\setminus V(G')$ with vertices of $M'\setminus D$ to get a subset $\widetilde{D}\subseteq V(G')$ that is an mds of $G$ and thus of $G'$.
	This means $\vw(\widetilde{D})=t$ and since $\vw$ is constant on $M$, we get $\vw(D)=t$.
\end{proof}

The first rule is about reducing a clique class to a certain number $r\in\IN$ of vertices.
\begin{center}
\parbox{0.9\columnwidth}{ \emph{$\boldsymbol{r}$-Clique Class Reduction}: If a clique class $C$ contains more than $r$ vertices, delete all but $r$ vertices from $C$.
}
\end{center}

\begin{lemma} \label{lem:cliqueclassreduction}
Let $G$ be a graph, $r,\,k\in \IN$ and $C\subseteq V(G)$ a clique class with $|C|>r$.
Furthermore, let $G'$ be the graph obtained from $G$ by applying the $r$-Clique Class Reduction rule with respect to $C$.
Then for all $t\leq r$, the graph $G$ is target-$t$ $k$-equidominating if and only if $G'$ is target-$t$ $k$-equidominating.
\end{lemma}
\begin{proof}
	The vertices of a clique class are pairwise mds-exchangeable, both in $G$ and $G'$.
	Further, we know that $|D\cap C|\leq 1$ for every mds $D\subseteq V(G)$ of $G$, by Observation~\ref{obs:einknotenproCC}.
	We define $C'\coloneqq C\cap V(G')$.
	Let $D\subseteq V(G')$ be a dominating set of $G'$.
	The vertex (or vertices) of $D$ that dominates $C'$ also dominates the vertices of $C\setminus C'$ in $G$.
	So, $D$ is a dominating set of $G$ and we can apply Lemma~\ref{lem:reductions_general}. 
	This finishes the proof. 
\end{proof}

The next rule is about the previous defined stable set bundles.
As seen before there can be arbitrarily large stable set bundles in an equidominating graph.
However, we can reduce them to a suitable size.
Again a positive integer $r\in\IN$ specifies the reduction rule.
\begin{center}
\parbox{0.9\columnwidth}{ \emph{$\boldsymbol{r}$-Stable Set Bundle Reduction}: If a stable set bundle $\mathcal{S}$ contains more than $r$ stable set classes, delete all but $r$ stable set classes of $\mathcal{S}$.
}
\end{center}

The following lemma shows that the $r$-Stable Set Bundle Reduction rule can be used to construct kernels for the parameterized problems.

\begin{lemma} \label{lem:stablesetbundlereduction}
	%Let $G$ be a graph, $r,\,k\in \IN$ and $\mathcal{S}=\{S_1, \ldots, S_N\}$ a stable set bundle with $N>r$ stable set classes.
	Let $G$ be a graph, $r,\,k\in \IN$ and $\mathcal{S}\subseteq V(G)$ a stable set bundle containing more than $r$ stable set classes.
	Further, let $G'$ be the graph obtained from $G$ by applying the $r$-Stable Set Bundle Reduction rule with respect to $\mathcal{S}$.
	Then for all $t\leq 2r$, the graph $G$ is target-$t$ $k$-equidominating if and only if $G'$ is target-$t$ $k$-equidominating.
\end{lemma}
\begin{proof}
	Note that there are $2r$ vertices in $\mathcal{S}\cap V(G')$.
	First, we consider the case $r=1$.
	In this case, $\mathcal{S}$ becomes a stable set class in $G'$.
	Nevertheless, the lemma holds. 
	There are only the two possible values for $t$: $t=1$ or $t=2$.
	Neither $G$ nor $G'$ are target-1 $k$-equidominating as the complete graphs $K_n$ ($n\in\IN$) are the only target-1 $k$-equidominating graphs (with equidominating structure $(\vw\equiv 1,1)$).
	
	For $t=2$, there can exist only one more pseudo class besides $\mathcal{S}$, otherwise $G$ and $G'$ are not target-2 $k$-equidominating, by Corollary~\ref{cor:differentpseudodifferentweights}.
	If this second pseudo class is an adjacent singleton class or an adjacent clique class, then both $G$ and $G'$ are target-2 $k$-equidominating.
	In the other cases neither $G$ nor $G'$ are target-2 $k$-equidominating since then an mds exists with more than two vertices.
	
	Secondly, let $r\geq 2$.
	Again, we show that all conditions of Lemma~\ref{lem:reductions_general} are met.
	Following Lemma~\ref{lem:SSCmitgleichemgewicht}, the vertices of $\mathcal{S}$ are pairwise mds-exchangeable.
	The same holds for the vertices of $\mathcal{S}\cap V(G')$ since at least two stable set classes of $\mathcal{S}$ remain in $G'$.
	It is easy to see that $|D\cap\mathcal{S}|\leq 2$ for every mds $D\subseteq V(G)$ of $G$.
	Finally, we can show, analogously to the proof of Lemma~\ref{lem:cliqueclassreduction}, that every mds of $G'$ is a dominating set of $G$.
\end{proof}

The last reduction rules considers clique bundles.
As we already know, the vertices of a clique bundle can have different neighborhoods. 
Thus, even though we can bound the number of pseudo classes of a graph, there can be a large number of twin classes (more precisely, singleton classes and clique classes) in a clique bundle.
In the case of more than one clique bundle being in a graph, there can be arbitrarily many distinct neighborhoods, and consequently, arbitrarily many twin classes in a clique bundle.
For example, the graph shown in Figure~\ref{fig:cliquebundle} on page~\pageref{fig:cliquebundle} can be extended to any number of twin classes in both clique bundles:
we just add two adjacent twin classes to the clique bundles $\mathcal{C}$ and $\mathcal{C}'$, respectively, analogously to the existing ones.
Furthermore, a special case can occur:
if a clique bundle only contains singleton classes, then more than one vertex of such a clique bundle can be in an mds.

These two facts make it harder to bound the number of vertices in a clique bundle in terms of the parameters $k$ and $t$ as in the case of clique classes or stable set bundles.
The keynote to overcome this is gathering the vertices of a clique bundle into certain subsets.
In such a subset, the vertices have the same neighborhood regarding all pseudo classes except clique bundles. 
With respect to clique bundles, however, the vertices of a subset have the same number of neighbors in each clique bundle.

To formalize the above-mentioned idea, we introduce the following notion.

\begin{definition}\label{def:mds-vector}
	Let $G$ be a graph with pseudo class partition $\{P_1,\ldots,P_s\}$.
	For every vertex $v\in V(G)$ we define the vector $\mu^v=(\mu^v_1,\ldots,\mu^v_s) \in \IN^s_0$ as follows: 
	if $v\in P_i$, then we set 
	\begin{subnumcases}{\label{def:mu_selbst-drin} \mu^v_i\coloneqq}
		 1, 		 & if $P_i$ is a singleton class, clique class or clique bundle, \label{def:mu_selbst-drin_SC-CC-CB}  \\
		 2, 		 & if $P_i$ is a stable set bundle, \label{def:mu_selbst-drin_SSB} \\
		 |P_i|,  & if $P_i$ is a stable set class. \label{def:mu_selbst-drin_SSC} 
	\end{subnumcases}
	If $v\notin P_i$ and $P_i$ is not a clique bundle, then we set
	\begin{subnumcases}{\label{def:mu_nicht-drin_noCB} \mu^v_i\coloneqq}
		 1, 											& if $v$ is  adjacent to $P_i$, \label{def:mu_nicht-drin_noCB_non-adjacent} \\
		 0, 											& if $v$ is not adjacent to $P_i.\qquad$ \label{def:mu_nicht-drin_noCB_adjacent}
	\end{subnumcases}	
	If $v\notin P_i$ and $P_i$ is a clique bundle, then we set
	\begin{subnumcases}{\label{def:mu_nicht-drin_CB} \mu^v_i\coloneqq}
		 |P_i\setminus N[v]|+1, & if there exists an mds $D\subseteq V$ with $N[v]\cap D\subseteq P_i$, \label{def:mu_nicht-drin_CB_mds_exists} \\
		 0, 										& otherwise. \label{def:mu_nicht-drin_CB_no_mds_exists}
	\end{subnumcases}
	We call $\mu^v$ the \emph{\vectorname-vector} of $v$.
\end{definition}
We remark that $\mu^v$ rather contains information about how $v$ can be dominated by the pseudo classes of a graph, than how $v$ dominates the pseudo classes (in particular in the cases (\ref{def:mu_nicht-drin_CB_mds_exists}) and (\ref{def:mu_nicht-drin_CB_no_mds_exists})).

Let $P_i$ be a clique bundle and $D$ an mds such that $N[v]\cap D\subseteq P_i$.
Then, $v$ is dominated only by vertices of $P_i$.
Furthermore, the vertices of $P_i$ are pairwise mds-exchangeable.
This means that there are more vertices in $D\cap P_i$ than in $P_i\setminus N[v]$.
So, the number $\mu^v_i$ tells us how many vertices of $P_i$ must be at least in an mds to dominate $v$, such that $v$ is dominated only by vertices of $P_i$.

To decide whether an mds $D$ with $N[v]\cap D\subseteq P_i$ exists, it is sufficient to check if the set $\left(V(G)\setminus N[v]\right)\cup P_i$ is a dominating set. 
If so, this dominating set contains an mds $D\subseteq \left(V(G)\setminus N[v]\right)\cup P_i$ with $N[v]\cap D\subseteq P_i$.

The values of $\mu^v$ are bounded by $t$ in every target-$t$ equidominating graph.

\begin{lemma}\label{lem:mu>r_not_equidom} 
	Let $G=(V,E)$ be a graph with pseudo class partition $\{P_1,\ldots,P_s\}$ and let $r\in\IN$.
	If there is a vertex $v\in V$ with $\mu^v_i>r$ for some $i\in [s]$, then $G$ is not target-$t$ equidominating for all $t\leq r$.
\end{lemma}
\begin{proof}
	Let $t\leq r$ and $\mu^v_i>r$.
	If $v\in P_i$ and $P_i$ is a stable set bundle, then $\mu^v_i=2$ and hence $r=1$. 
	Since we can extend two non-adjacent vertices of $P_i$ to an mds (with at least two elements), $G$ is not target-1 equidominating.
	Analogously, if $v\in P_i$ and $P_i$ is a stable set class, there exists an mds containing more than $t$ vertices.
	
	The last possible case for $\mu^v_i>1$ is $v\notin P_i$ and $P_i$ is a clique bundle.
	Then, there exists an mds $D$ with $N[v]\cap D\subseteq P_i$ and $|D|\geq|D\cap P_i|\geq \mu^v_i > t$.
	Again, $G$ is not target-$t$ equidominating.
\end{proof}

As mentioned before, we want to gather -- or rather partition -- the vertices of clique bundles into subsets.
We do this in a way such that the vertices of each subset have identical \vectorname-vectors.
However, a problem arises if we then reduce such a subset of a graph.
Likewise for the other reduction rules, the main condition to prove that the reduction is safe for the parameterized problems is: 
a subset of vertices of the reduced graph is an mds if and only if it is an mds of the original graph (compare Lemma~\ref{lem:reductions_general}(iii)). 

Now the problem is the following:
it can occur that the pseudo class partition changes if we delete some vertices of a clique bundle.
Clearly, if the pseudo class partition changes, then we also obtain different \vectorname-vectors.
With possibly different pseudo class partitions and \vectorname-vectors we cannot prove the above-mentioned main condition.

Since the clique bundles are determined by the graph, we need a more general structure.
Therefore, we introduce a new mathematical object and transfer all relevant information and properties to it.
However, we do this in a way such that neither the pseudo class partition nor the \vectorname-vectors change when we delete some elements from it.
Roughly speaking, the graph does not determine the partition, but the partition determines the graph.

\begin{definition}\label{def:pseudo_graph}
	A \emph{pseudo graph} is a triple $\mathfrak{P}=(V,\mathcal{P},\mu)$, consisting of a non-empty set $V$, a partition $\mathcal{P}=\{P_1,\ldots,P_s\}$ of $V$ and a function $\mu \colon V \to \N_0^s$, such that $\mu(v)\neq (0,\ldots,0)$ for all $v\in V$.
\end{definition}

Next, we need an equivalent for (minimal) dominating sets in pseudo graphs.
By $(\mu(v))_i$ we denote the $i$-th component of $\mu(v)$.

\begin{definition}\label{def:mu-feasible}
	Let $\mathfrak{P}=(V,\mathcal{P}=\{P_1,\ldots,P_s\},\mu)$ be a pseudo graph and $X\subseteq V$.
	We call $X$ a \emph{dense set of $\mathemph{\mathfrak{P}}$} if for every $v\in V$ a block $P_i\in\mathcal{P}$ exists with $0<(\mu(v))_i\leq|X\cap P_i|$.
	If every proper subset $X'\subsetneq X$ is not dense, then $X$ is said to be a \emph{minimal dense set of $\mathemph{\mathfrak{P}}$}.
\end{definition}

If it is not required for each $v\in V$ that at least one component of $\mu(v)$ is greater than zero (see Definition~\ref{def:pseudo_graph}), then it is possible that no (minimal) dense set exists at all.
The next definition motivates the two previous definitions.

\begin{definition}
	Let $G=(V,E)$ be a graph with pseudo class partition $\{P_1,\ldots,P_s\}$ and \vectorname-vectors $\mu^v\in\IN_0^s$ for $v\in V$.
	By $\mathfrak{P}(G)=(V,\{P_1,\ldots,P_s\},\mu)$ we denote the \emph{pseudo graph of $\mathemph{G}$}, where $\mu(v)\coloneqq \mu^v$ for each $v\in V$.
	We say that a pseudo graph $\mathfrak{P}$ is \emph{induced by} a graph $G$ if $\mathfrak{P}(G)=\mathfrak{P}$.
\end{definition}

Note that the pseudo graph of a graph is unique (up to the order of the pseudo class partition), while two different graphs can have the same pseudo graph.
For example, two graphs have identical pseudo graphs if they differ only with respect to edges between a clique bundle $\mathcal{C}$ and a vertex $v$ where there is no mds $D$ with $N[v]\cap D \subseteq \mathcal{C}$.
Further, there are pseudo graphs that are not induced by a graph. % zb wenn Symmetrie verletzt.

We can compute the pseudo class partition in time $\mathcal{O}(nm^2)$.
Furthermore, we can determine all \vectorname-vectors in time $\mathcal{O}(n^2 + nm)$.
Hence, the pseudo graph of a graph with $n$ vertices and $m$ edges can be computed in time $\mathcal{O}(nm^2+n^2)$. 

The next lemma and the subsequent corollary show that induced pseudo graphs and minimal dense sets indeed correspond to mds as desired.

\begin{lemma}\label{lem:ds_iff_mu-feasible}
	Let $G=(V,E)$ be a graph with \vectorname-vectors $\mu^v$, $v\in V$, pseudo class partition $\mathcal{P}=\{P_1,\ldots,P_s\}$ and pseudo graph $\mathfrak{P}(G)$.
	Further, let $D\subseteq V$.
	Then $D$ is a dominating set of $G$ if and only if $D$ is a dense set of $\mathfrak{P}(G)$.
\end{lemma}
\begin{proof}
	\uline{$\Longrightarrow$:}
	Let $D\subseteq V$ be a dominating set of $G$ and $v\in V$ with $v\in P_l$ ($l\in[s]$). 
	If $P_l$ is a stable set class, then every vertex of $P_l$ or a neighbor $x\in P_i\cap N(v)$ lies in $D$.
	In the first case we have $0<\mu_l^v=|P_l|=|D\cap P_l|$ (see~(\ref{def:mu_selbst-drin_SSC}) in Definition~\ref{def:mds-vector}).
	In the latter case, if there is a vertex $w\in N[v]\cap D$ that lies in a singleton class, a clique class, a stable set class or a stable set bundle $P_i$ ($i\in[s]$), then we have $0<\mu_i^v=1\leq|D\cap P_i|$, by~(\ref{def:mu_nicht-drin_noCB_adjacent}).
	If each vertex of $N[v]\cap D$ lies in a clique bundle, then we cannot exchange (within each clique bundle) all neighbors of $v$ lying in $D$ with non-neighbor of $v$ not lying in $D$, (this would contradict the mds-exchangeability).
	This means that there is a clique bundle $P_{i}$ with $0<\mu^v_i\leq|D\cap P_i|$ (compare~(\ref{def:mu_nicht-drin_CB_mds_exists})).
	
	If $P_l$ is a singleton class, a stable set class, a clique class or a clique bundle, we can show analogously that $0<\mu_i^v=|D\cap P_i|$ holds for some $i\in[s]$, by considering (\ref{def:mu_selbst-drin_SC-CC-CB}) and (\ref{def:mu_selbst-drin_SSB}).
	Hence, $D$ is a dense set of $\mathfrak{P}(G)$.

	\uline{$\Longleftarrow$:}
	Let $D$ be a dense set of $\mathfrak{P}$ and $v\in V$ with $v\in P_l$ ($l\in[s]$).	
	Then, there exists $P_i\in \mathcal{P}$ with $0<(\mu(v))_{i}\leq|D\cap P_i|$.
	If $i=l$, then $v$ or a vertex of $N(v)\cap P_l$ lies in $D$, by (\ref{def:mu_selbst-drin}).
	Thus, $v$ is dominated.
	If $i\neq l$ and $P_{i}$ is a singleton class, stable set class, clique class or a stable set bundle (in $G$), then $v$ is dominated by the vertices of $D\cap P_{i}$, by (\ref{def:mu_nicht-drin_noCB_adjacent}).
	If $P_{i}$ is a clique bundle, then $v$ has less non-neighbors in $P_{i}$ than there are vertices in $D\cap P_{i}$ (since $|P_i\setminus N[v]|+1\leq |D\cap P_{i}|$, by (\ref{def:mu_nicht-drin_CB_mds_exists})).
	This means, at least one vertex of $N(v)\cap P_{i}$ lies in $D$.
	Hence, $v$ is dominated.
	It follows that $D$ is a dominating set of $G$ and the proof is finished.
\end{proof}

Using Lemma~\ref{lem:ds_iff_mu-feasible}, it is straightforward to prove the following corollary by contradiction.

\begin{corollary}\label{cor:mds_iff_min-mu-feasible}
	Let $G=(V,E)$ be a graph with pseudo graph $\mathfrak{P}(G)$ and let $D\subseteq V$.
	Then $D$ is an mds of $G$ if and only if $D$ is a minimal dense set of $\mathfrak{P}(G)$.
\end{corollary}

Next, we introduce analogous terms to identify the minimal dense sets of a pseudo graph.
\begin{definition}
	A pseudo graph $\mathfrak{P}=(V,\mathcal{P},\mu)$ is called \emph{equidense} if there exists $t\in\IN$ and a weight function $\vw \colon V \rightarrow \IN$ such that for all $D\subseteq V$ the following equivalence holds: 
	\[D \text{ is a minimal dense set} \iff \vw(D) = t.\] 
	Further, we call the pair $(\vw,t)$ an \emph{equidense structure}, $\vw$ an \emph{equidense function} and $t$ a \emph{target value}. 
\end{definition}

\begin{definition}\label{def:k-equi-mu-feasible}
	For a given $t\in\IN$, a pseudo graph $\mathfrak{P}=(V,\mathcal{P},\mu)$ is called \emph{target-$\mathemph{t}$ equidense} if there exists an equidense structure of the form $(\vw,t)$ of $G$. 
\end{definition}

\begin{definition}\label{def:target-t_equi-mu-feasible}
	For a given $k\in\IN$, a pseudo graph $\mathfrak{P}=(V,\mathcal{P},\mu)$ is said to be \emph{$\mathemph{k}$-equidense} if there exists an equidense structure $(\vw,t)$ with $\vw \colon V \rightarrow [k]$ for some $t\in\IN$.
	In this case, $(\vw,t)$ is called a \emph{$\mathemph{k}$-equidense structure} and $\vw$ a \emph{$\mathemph{k}$-equidense function}.
\end{definition}

Finally, we call a pseudo graph \emph{target-$\mathemph{t}$ $\mathemph{k}$-equidense} if a $k$-equidense structure with target value $t$ exists.
By Corollary~\ref{cor:mds_iff_min-mu-feasible}, we immediately get the following result.

\begin{corollary}\label{cor:equidom_iff_equi-mu-feasible}
	Let $G$ be a graph with pseudo graph $\mathfrak{P}(G)$ and let $k,\,t\in\IN$.
	Then $G$ is target-$t$ $k$-equidominating if and only if $\mathfrak{P}(G)$ is target-$t$ $k$-equidense.
	Moreover, in the affirmative case, we can use the same structure to identify minimal dominating and minimal dense sets in $G$ and $\mathfrak{P}(G)$, respectively.
\end{corollary}

Now, we gathered together everything to define the reduction rule and to prove that we can use it (in combination with Corollary~\ref{cor:equidom_iff_equi-mu-feasible}) for the {\sc Target-$t$ Equidomination} problem and the $k$-{\sc Equidomination} problem. 
Again, the rule is specified by a positive integer $r\in\IN$.

\begin{center}
\parbox{0.9\columnwidth}{\emph{$\boldsymbol{r}$-Pseudo Graph Reduction}: If a subset $M \subseteq P$ of a block $P$ of the partition of a pseudo graph with $\mu(v)=\mu(w)$ for all $v,\,w\in M$ contains more than $r$ vertices, delete all but $r$ vertices of $M$.
}
\end{center}

\begin{lemma} \label{lem:clique_bundle_subset_reduction}
	Let $r,\,k\in \IN$ and $\mathfrak{P}=(V,\mathcal{P}=\{P_1,\ldots,P_s\},\mu)$ be a pseudo graph induced by a graph $G=(V,E)$, with $(\mu(v))_i\leq r$ for all $v\in V$ and $i\in[s]$.
	Let $P_l$ ($l\in[s]$) be a block of $\mathcal{P}$ such that $P_l$ is a clique bundle of $G$.
	Further, let $M \subseteq P_l$ be a subset of  $P_l$ with $\mu(v)=\mu(w)$ for all $v,\,w\in M$ and $|M|>r$.
	
	Let $\mathfrak{P}'=(V',\mathcal{P}',\mu')$ be the pseudo graph obtained from $\mathfrak{P}$ by applying the $r$-Pseudo Graph Reduction rule with respect to $M$.
	Then for all $t\leq r$, $\mathfrak{P}$ is target-$t$ $k$-equidense if and only if $\mathfrak{P}'$ is target-$t$ $k$-equidense.
\end{lemma}
\begin{proof}
	Note that $\mu'=\mu\big|_{V'}$ and that besides $P_l$ the partitions $\mathcal{P}$ and $\mathcal{P}'$ of $V$ and $V'$, respectively, have identical blocks.
	
	In the first place, we claim that a subset $D\subseteq V'$ is a minimal dense set of $\mathfrak{P}'$ if and only if $D$ is a minimal dense set of $\mathfrak{P}$.
	To prove this, it is sufficient to show that the equivalence holds for dense sets (not necessarily minimal).
	So, let $D\subseteq V'$.
	First, let $D$ be a dense set of $\mathfrak{P}'$.
	Since $\mu(v)=\mu(w)$ for all $v\in P_l\setminus P'_l$ and $w\in P'_l$, we directly get that $D$ is a dense set of $\mathfrak{P}$.
	Secondly, let $D$ be a dense set of $\mathfrak{P}$.
	By considering the definition of dense, $D$ is clearly a dense set of $\mathfrak{P}'$ and the claim is proved.

	Now, we prove the equivalence stated in the lemma.
	Let $t\leq r$ and $M'= M\cap V'$.
	
	\uline{$\Longrightarrow$:} Let $\mathfrak{P}$ be target-$t$ $k$-equidense and $(\vw,t)$ be a $k$-equidense structure of $\mathfrak{P}$.
	Any subset of $V'$ is a minimal dense set of $\mathfrak{P}'$ if and only if it is a minimal dense set of $\mathfrak{P}$, which in turn is the case if and only if it has total weight $t$.	
	It follows that $(\vw',t)$ with $\vw'\coloneqq\vw\big|_{V'}$ is a $k$-equidense structure of $\mathfrak{P}'$ with target value $t$.

	\uline{$\Longleftarrow$:} Let $\mathfrak{P}'$ be target-$t$ $k$-equidense and $(\vw',t)$ be a $k$-equidense structure of $\mathfrak{P}'$.
	Since $\mathfrak{P}$ is induced by a graph and $P_l$ is a clique bundle of $G$, there exists a minimal dense set $D$ with $|D\cap P_l|=1$.
	It follows that $\vw'$ is constant on $P_l \cap V'$.
	We define 
	\[\vw(v)\coloneqq \begin{cases} \vw'(v), & \text{if }v \in  V',\\
																	\vw'(w), & \text{otherwise}, \end{cases}\]
	with any $w \in P_l \cap V'$ and claim that $(\vw,t)$ is a $k$-equidense structure of $\mathcal{P}$.
	It is clear that $\vw(v)\leq k$ for every $v\in V$.

	First, let $X\subseteq V$ be a subset with $\vw(X)=t$. 
	As $\vw(X\cap M)\leq\vw(X)=t$, we get $|X\cap M|\leq t\leq r$. 
	Since $|M'|=r$, we can assume $|X\cap M|\subseteq M'$ and hence $X\subseteq V'$.
	It follows that $X$ is a minimal dense set of $\mathfrak{P}'$ and thus also of $\mathfrak{P}$.
	
	Secondly, let $D\subseteq V(G)$ be a minimal dense set of $\mathfrak{P}$.
	Since $(\mu(v))_i\leq r$ for all $v\in V$ and $i\in[s]$, we have $|D\cap M|\leq|D\cap P_l|\leq r$.
	Again, we can assume that $D\subseteq V'$.
	This means that $D$ is a minimal dense set of $\mathfrak{P}'$. 
	It follows that $\vw'(D)=t$ and consequently $\vw(D)=t$.
	This finishes the proof.
\end{proof}

\section{A Kernel for the {\sc Target}-\textit{t} {\sc Equidomination} Problem} \label{sec:kerneltargett}

We will now prove that the {\sc Target-$t$ Equidomination} problem is fixed-parameter tractable. 
For that, we construct a generalized kernel the size of which can be bounded by a function of $t$.
We apply the XP algorithm described in Section~\ref{sec:xpalgorithm} to the kernel.
In this way, we do not only obtain a complexity result but also an explicit {FPT} algorithm.

\begin{theorem} \label{thm:targettequidominationfpt}
	The \textsc{Target}-$t$ \textsc{Equidomination} problem admits a generalized kernel of $\mathcal{O}\left(t^{t+1}\right)$ which is computable in polynomial time.
	Moreover, there is an algorithm to solve the \textsc{Target}-$t$ \textsc{Equidomination} problem that runs in time $\mathcal{O}\left(nm^2 +n^2 + t^{2t^2+3t+1}\right)$ for a graph on $n$ vertices and $m$ edges. 
\end{theorem}
\begin{proof}
	Let $G$ be a graph with $|V(G)|=n$, $|E(G)|=m$ and $t\in\IN$.
	First, we use an algorithm for modular decomposition to compute the twin partition.
	If there is a stable set class with more than $t$ vertices, then $G$ is not target-$t$ equidominating, by Observation~\ref{obs: SSC-maximal-t-knoten}.
	If no such stable set class exists, then we apply Algorithm~\ref{alg:adjacentmdsexchangeable} to obtain the pseudo class partition $\{P_1,\ldots,P_s\}$ and we determine the \vectorname-vectors $\mu^v$ for all $v\in V(G)$.
	
	If $s>t$ or $\mu^v_i>t$ for some $v\in V(G)$ and $i\in [s]$, then we conclude that $G$ is not target-$t$ equidominating, by Corollary~\ref{cor:differentpseudodifferentweights} and Lemma~\ref{lem:mu>r_not_equidom}.
	Otherwise, we apply the $r$-Clique Class Reduction rule with $r=t$ and the $r$-Stable Set Bundle Reduction rule with $r=\lceil t/2\rceil$ to all clique classes and stable set bundles, respectively, to obtain an induced subgraph $G'$ of $G$.
	Then, we determine the pseudo graph $\mathfrak{P}(G')$ and apply the $r$-Pseudo Graph Reduction rule (simultaneously) to all blocks of $\mathfrak{P}(G')$ that are clique bundles of $G'$, again with $r=t$.
	This yields a pseudo graph $\mathfrak{P}''=(V'',\mathcal{P}'',\mu'')$.
	
	Now, $\mathcal{P}''$ has at most $t$ blocks.
	Every block that is a clique class or stable set class in $G$ contains at most $t$ elements.
	Blocks of stable set bundles may have up to $t+1$ elements. 
	In blocks of clique bundles there can be at most $(t+1)^{(t-1)}$ distinct \vectorname-vectors (note that $\mu^v(i)=1$ for all vertices $v$ within a clique bundle $P_i$).
	Thus, after the reductions  every block of $\mathcal{P}''$ contains at most $t(t+1)^{(t-1)}$ elements.
	
	Taken together, we get $n''\coloneqq|V''|= \mathcal{O}\left(t^{t+1}\right)$.
	Finally, by Lemma~\ref{lem:cliqueclassreduction}, Lemma~\ref{lem:stablesetbundlereduction} and Lemma~\ref{lem:clique_bundle_subset_reduction}, the obtained pseudo graph $\mathfrak{P}''$ is target-$t$ equidense if and only if $G$ is target-$t$ equidominating.
	
	We can compute the pseudo class partition of $G$ in time $\mathcal{O}(nm^2)$.
	The computation of the \vectorname-vectors $\mu^v$ for all $v\in V(G)$ and the partitioning of the clique bundles in preparation for the Pseudo Class Reduction can be done in $\mathcal{O}(n^2 + nm)$.
	Applying the three reduction rules needs linear time.
	This finishes the proof of the first statement of this theorem.

	By the proof of Lemma~\ref{lem:reductions_general}, we know that we can easily extend any target-$t$ equidominating structure of $G'$ to a target-$t$ equidominating structure of $G$.
	The same holds for a target-$t$ equidense structure of $\mathfrak{P}'$ and $\mathfrak{P}$.
	Moreover, we can use the same function and target value for $G'$ and $\mathfrak{P}'(G')$, by Corollary~\ref{cor:equidom_iff_equi-mu-feasible}.
	It follows that an equidense structure of $\mathfrak{P}'$ induces an equidominating structure of $G$.
	
	Furthermore, by the proofs of the lemmas of Section~\ref{sec:reduction}, we know that a subset $D\subseteq V''$ is a minimal dense set of $\mathfrak{P}''$ if and only if $D$ is an mds of $G$. 
	Due to this equivalence, we can analogously apply Algorithm~\ref{alg:xp_algorithm} to obtain a target-$t$ equidense structure of $\mathfrak{P}''$, if existent.
	
	Since we do not have to compute the pseudo class partition, the first two summands of the running time Algorithm~\ref{alg:xp_algorithm} vanish.
	So, for applying Algorithm~\ref{alg:xp_algorithm} to $\mathfrak{P}''$ we need time 
	\begin{align*}
		\mathcal{O}\left({n''}^tt^t+{n''}^{2t+2}t^{-t-1}+t^{3t+3}\right)
			& =\mathcal{O}\left(t^{t^2+2t}+t^{2t^2+4t+2}t^{-t-1}+t^{3t+3}\right) \\
			& =\mathcal{O}\left(t^{2t^2+3t+1}\right).
	\end{align*}
	
	Together with the computation of the pseudo class partition and the \vectorname-vectors of $G$, we achieve a total running time of
	\[\mathcal{O}\left(nm^2 + n^2 + nm + t^{2t^2+3t+1}\right)=\mathcal{O}\left(nm^2 +n^2 + t^{2t^2+3t+1}\right).\]
\end{proof}

\section{A Kernel for the \textit{k}-{\sc Equidomination} Problem} \label{sec:kernelkequidomination}

In this section we show that the $k$-\textsc{Equidomination} problem admits a generalized kernel the size of which is bounded by a function of $k$ and, therefore, is FPT.
The proof of Theorem~\ref{thm:kequidominationfpt} is along the lines of the argumentation used in \cite{kequistable_fpt} to prove that the $k$-\textsc{Equistable} problem is FPT.
In the following, we assume that $k> 1$.
The only graphs that are 1-equidominating are the graphs $K_n$, $\overline{K_n}$ and $T(2n,n)$ ($n\in\IN$) with equidominating structure $(\vw\equiv 1,1)$, $(\vw\equiv 1,n)$ and $(\vw\equiv 1,2)$, respectively. 

\begin{theorem} \label{thm:kequidominationfpt}
	The $k$-\textsc{Equidomination} problem admits a generalized kernel of size $\mathcal{O}(k^{3k+1})$ which is computable in polynomial time. 
	Furthermore, there is an algorithm to solve the $k$-\textsc{Equidomination} problem which runs in time $\mathcal{O}\left(nm^2 + n^2 + k^{6k^2+7k+1}\right)$ for a given graph on $n$ vertices and $m$ edges.
\end{theorem}
The proof of this theorem builds upon two more lemmas. 

\begin{lemma} \label{lem:onlyonebigssc}
	A graph $G=(V,E)$ is not $k$-equidominating ($k\in \IN$) if $G$ has two pseudo classes of size at least $k^2$, where one of those pseudo classes is a stable set class.
\end{lemma}
\begin{proof}
	For $k\in \IN$, let $G=(V,E)$ be a graph with a stable set class $S\subseteq V$ and a different pseudo class $P\subseteq V$ such that $\min\{|S|,|P|\}\geq k^2$.
	Suppose that $G$ is $k$-equidominating with $k$-equidominating function $\vw \colon V \rightarrow [k]$.
	It is straightforward that -- regardless of what kind of pseudo class $P$ is -- there is an mds $D$ with $S\subseteq D$ and $|P\cap D|\leq 2$.

	Let $i,j\in [k]$ be weights such that $|\{s\in S\mid\vw(s)=i\}|\geq k$ and $|\{p\in P\mid\vw(p)=j\}|\geq k$.
	Such weights exist due to the size of $S$ and $P$.
	Further, there are at least $k$ vertices of $P$ of weight $j$ are not in $D$.
	Now, let $S'\subseteq S$ be a subset of $j$ vertices of weight $i$ and $P'\subseteq P$ be a subset of $i$ vertices of weight $j$ with $P'\cap D=\emptyset$.
	Consequently we get $\vw(S')=\vw(P')$.
	This leads to a contradiction as the set $(D\setminus S')\cup P'$ can not be an mds (cf.\ Observation~\ref{obs:keinereineroderalleinSSC}) while being of the same weight as the mds $D$.
\end{proof}

\begin{lemma} \label{lem:konstantaufisolierterssc}
	Let $G=(V,E)$ be a graph and $k\in \IN$.
	Furthermore, let $S$ be an isolated stable set class of size at least $k^5$ and $|S\setminus V|\leq k^3$.
	Then $G$ is $k$-equidominating if and only if there exists a $k$-equidominating function that is constant on $S$.
\end{lemma}
\begin{proof}
	First note that the stable set class $S$ is contained in every mds of $G$.
	The sufficiency-part is trivial.
	
	So, let $G$ be $k$-equidominating with $k$-equidominating structure $(\vw,t)$.
	As there are only $k$ different weights and $|S|\geq k^5$, a weight $i\in [k]$ exists such that $|S_i|\geq k^4$, with $S_i\coloneqq\{s\in S\mid\vw(s)=i\}$.
	We show that $(\vw',t')$ with 
	\[\vw'(v)\coloneqq \begin{cases} \vw(v), & \text{if } v \in  V\setminus S,\\
																			i, 	 & \text{if } v\in S, \end{cases}\]
	and $t'\coloneqq t-\vw(S)+\vw'(S)=t-\vw(S)+i|S|$ is a $k$-equidominating structure.
	It is easy to see that $\vw'$ is bounded by $k$ and that $\vw'(D)=t'$ holds for every mds $D$.

	As before, the only tricky part is to show that any subset $X\subseteq V$ with $\vw'(X)=t'$ is an mds.
	So, let $X\subseteq V$ be a subset of vertices with $\vw'(X)=t'$.
	
	We define $r\coloneqq|S\setminus X|$ and suppose that $r> k^4$.
	Since $|X\setminus S|\leq |V\setminus S| \leq k^3$, we get $\vw'(X\setminus S)\leq k^4$.
	It follows that 
	\begin{align*}
		\vw'(X) & = \vw'(S)-\vw'(S\setminus X)+\vw'(X\setminus S)  \\
						& \leq i(|S|-r)+k^4 \\
						& < i(|S|-r) +ir \\
						&  = i|S| \\
						& \leq t',
	\end{align*}
	a contradiction. 
	So, $r=|S\setminus X|\leq k^4$ and hence, together with $k^4\leq|S_i|$ and $\vw'\big|_{S}\equiv i$, we may assume that $S\setminus X \subseteq S_i$.
	Otherwise we can exchange vertices of $(S\setminus X)\setminus S_i$ with vertices of $X\cap S_i$ until $S\setminus X \subseteq S_i$.
	In doing so, we maintain the total weight $t'$ as well as the property whether being an mds.
	This yields 
	\begin{align*}
		\vw(X) & = \vw'(X) -\vw'(X\cap S) + \vw(X\cap S)  \\
							& = t' - i(|S|-r) + \vw(X\cap S) \\
							& = t' - i|S| + (\vw(X\cap S) +ir) \\
							& = t' - \vw'(S) + (\vw(X\cap S) + \vw(S\setminus X)) \\
							& = t' - \vw'(S) + \vw(S) \\
							& = t.
	\end{align*}
	Hence, $X$ is an mds of $G$ and the proof is finished.
\end{proof}
 
It follows that if we want to check a graph of the form like the one given in Lemma~\ref{lem:konstantaufisolierterssc} for $k$-equidomination, it is sufficient to work with $k$-equidominating functions that are constant on the isolated stable set class.
Now, we gathered all tools to prove Theorem~\ref{thm:kequidominationfpt}.
\begin{proof}[Proof of Theorem~\ref{thm:kequidominationfpt}.]
	Let $G$ be a graph with $|V(G)|=n$, $|E(G)|=m$ and $k\in\IN$.
	First, we decompose $G$ into pseudo classes.
	If there are more than $k$ pseudo classes, then $G$ is not $k$-equidominating, by Corollary~\ref{cor:differentpseudodifferentweights}.
	Otherwise, following Lemma~\ref{lem:onlyonebigssc}, two cases may occur.
	If none of the cases are fulfilled, then $G$ is also not $k$-equidominating.
	\begin{case}
		Every pseudo class of size at least $k^2$, if any, is either a clique class, a clique bundle or a stable set bundle.
	\end{case}
	As far as existent, we take one vertex of every singleton class, clique class and clique bundle, two vertices of an isolated stable set bundle, all vertices of an isolated stable set class and one neighbor of every other stable set class and stable set bundle.
	By this, we get a dominating set of size at most $k^2+k$, which contains an mds.
	This means $t\leq k^3+k^2$ must hold for every $k$-equidominating structure $(\vw,t)$ of $G$.
	Then, we compute all \vectorname-vectors and check if one of them has a component greater than $k^3+k^2$.
	If so, $G$ is not $k$-equidominating, by Lemma~\ref{lem:mu>r_not_equidom}.
	
	Otherwise, we perform the $r$-Clique Class Reduction rule and the $r$-Stable Set Bundle Reduction rule with $r = k^3+k^2$ and $r = \lceil(k^3+k^2)/2\rceil$, respectively. 
	Next, we compute the pseudo graph of the reduced graph and apply the $r$-Pseudo Graph Reduction rule with $r = k^3+k^2$.
	We obtain a pseudo graph $\mathfrak{P}''=(V'',\mathcal{P}'',\mu'')$ with $|V''|=\mathcal{O}(k^{3k+1})$ (compare the proof of Theorem~\ref{thm:targettequidominationfpt}).
	Again, by Lemma~\ref{lem:cliqueclassreduction}, Lemma~\ref{lem:stablesetbundlereduction} and Lemma~\ref{lem:clique_bundle_subset_reduction}, the obtained pseudo graph $\mathfrak{P}''$ is $k$-equidense if and only if $G$ is $k$-equidominating.
	\begin{case}
		There is a unique stable set class $S$ with $|S|\geq k^2$ and every other pseudo class has fewer than $k^2$ vertices.
	\end{case}
	In this case $|V(G)\setminus S| \leq k^3$, since there are at most $k-1$ pseudo classes besides $S$ and each of them has fewer than $k^2$ vertices.
	Further, we can assume that $|S|\geq k^5$ since otherwise $|V(G)|\leq k^5+k^3= \mathcal{O}(k^{3k+1})$ and the proof is finished.
	We distinguish two subcases.

	\noindent\textit{Case 2.1.~The stable set class $S$ does not see any other twin class.}\newline	
	We now construct a graph $G'$ by deleting all but $k^5$ many vertices of $S$ and claim that $G'$ is $k$-equidominating if and only if $G$ is $k$-equidominating.
	Let $S'\coloneqq S\cap V(G')$ be the set of the remaining vertices of $S$.
	Note that $S'$ and $S$ are contained in every mds of $G'$ and $G$, respectively.
	 
	\uline{$\Longleftarrow$:} Let $G$ be $k$-equidominating and $(\vw,t)$ a $k$-equidominating structure.
	We show that $(\vw',t')$ with $\vw'\coloneqq\vw\big|_{V(G')}$ and $t'\coloneqq t-\vw(S\setminus S')$ is a $k$-equidominating structure of $G'$.
		
	Let $D'$ be an mds of $G'$. 
	Then, $D'\disjcup (S\setminus S')$ is an mds of $G$ and consequently $\vw'(D')=\vw(D')=t-\vw(S\setminus S')=t'$.
	
	Now, let $X'\subseteq V(G')$ be a subset of vertices with $\vw'(X')=t'$.
	Then, $\vw(X'\disjcup(S\setminus S'))=t$ and hence $X'\disjcup(S\setminus S')$ is an mds of $G$.
	It follows that $X'$ is an mds of $G'$.
	
	\uline{$\Longrightarrow$:}  Let $G'$ be $k$-equidominating with $k$-equidominating structure $(\vw',t')$.
	Since all conditions of Lemma~\ref{lem:konstantaufisolierterssc} are met, we may assume that $\vw'\big|_{S'}\equiv i$ for some $i\in [k]$.
	We prove that $(\vw,t)$ with 
	\[\vw(v)\coloneqq \begin{cases} \vw'(v), & \text{if }v \in  V(G'),\\
																		 i,    & \text{if }v\in S\setminus S', \end{cases}\]
	and $t\coloneqq t'+i|S\setminus S'|$ is a $k$-equidominating structure of $G$. 

	Let $D\subseteq V(G)$ be an mds of $G$.
	Then there is an mds $D'\subseteq V(G')$ of $G'$ such that $D=(S\setminus S')\disjcup D'$.
	Thus $\vw(D)=\vw(S\setminus S')+\vw(D')=i|S\setminus S'|+\vw'(D')=t$. 

	Now let $X\subseteq V(G)$ be a subset $\vw(X)= t$.
	By $|V(G)\setminus S|\leq k^3$, we get that $\vw(V(G)\setminus S)\leq k^4$.
	Suppose $|S\setminus X|>k^4$.
	Then
	\begin{align*}
		\vw(X) & = \vw(X\cap S) + \vw(X\cap(V(G)\setminus S)) \\
							& \leq \vw(S) - \vw(S\setminus X) + \vw(V(G)\setminus S) \\
							& \leq i|S| - i|S\setminus X| + k^4 \\
							& < i|S| - ik^4 + k^4 \\ 
							& \leq i|S| \\
							& = \vw'(S') + i|S\setminus S'| \\
							& \leq t'+ i|S\setminus S'| \\
							& = t,
	\end{align*}
	a contradiction.
	Thus, $|S\setminus X|\leq k^4$ and together with $k^4 \leq |S'|$ we may assume that $S\setminus X \subseteq S'$ (compare proof of Lemma~\ref{lem:konstantaufisolierterssc}).
	With $X'\coloneqq X\cap V(G')$ we get $S\setminus S'= X\setminus X'$.
	It follows that $\vw'(X')=\vw(X)-i|X\setminus X'|=\vw(X)-i|S\setminus S'|=t'$. %Es gilt hier S\S'=X\X'
	So, $X'$ is an mds of $G'$ and consequently $X$ is an mds of $G$.
	Hence, $G$ is $k$-equidominating and the claim is proved.
	
	Taken together, we have proved that it is sufficient to check whether $G'$ is $k$-equidomi\-na\-ting with $|V(G')|\leq k^5+k^3 = \mathcal{O}(k^{3k+1})$. 

	\noindent\textit{Case 2.2. The stable set class $S$ sees another twin class $T$.}\newline
	So, suppose $G$ is $k$-equidominating with $k$-equidominating structure $(\vw,t)$ and let $D$ be an mds with $D \cap S=\emptyset$.
	It follows that $|D|\leq |V(G)\setminus S|\leq k^3$ and therefore $t \leq k^4$.
	But at the same time $t\geq\vw(S)\geq k^5$, a contradiction.
	Hence, $G$ is not $k$-equidominating.
	
	The determination of the pseudo class partition, the \vectorname-vectors and the pseudo graph as well as the application of the reduction rules can be done in polynomial time (compare the proof of Theorem~\ref{thm:targettequidominationfpt}).
	This finishes the proof of the first assertion of this theorem.

	Having the pseudo class partition and the \vectorname-vectors at hand, we conclude that $G$ is not $k$-equidominating if there are more than $k$ pseudo classes, two stable set classes of size at least $k^2$ or a non-isolated stable set class of size at least $k^5$.
	If no stable set class of size at least $k^2$ exists, we first apply the three reduction rules in linear time and then Algorithm~\ref{alg:xp_algorithm} to the obtained pseudo graph $\mathfrak{P}''$ of size $n''=\mathcal{O}(k^{3k+1})$ in time
	\begin{align*}
		\mathcal{O}\left({n''}^k k^k + {n''}^{2k+2} k^{-k-1} + k^{3k+3}\right)
			& =\mathcal{O}\left(k^{3k^2+2k}+k^{6k^2+7k+1}+k^{3k+3}\right) \\
			& =\mathcal{O}\left(k^{6k^2+7k+1}\right).
	\end{align*}
	If an isolated stable set class exists, then we apply Algorithm~\ref{alg:xp_algorithm} to the graph $G'$ on at most $k^5+k^3$ vertices, obtained by reducing the stable set class to $k^5$ vertices (if necessary).

	This yields a total running time (compare proof of Theorem~\ref{thm:targettequidominationfpt}) of
	\[\mathcal{O}\left(nm^2 + n^2 + k^{6k^2+7k+1}\right).\]
\end{proof}

\section{Hereditarily Equidominating Graphs} \label{sec:hereditarily}

A graph $G$ is called \emph{hereditarily equidominating} if every induced subgraph of $G$ is equidominating.
In this section, we give a characterization of the class of hereditarily equidominating graphs in terms of the list of forbidden induced subgraphs and a structural decomposition.
This decomposition yields an $\mathcal O (n (n+m))$ time recognition algorithm.

In order to state our characterization, we need some more notions.
A \emph{chain graph} is a bipartite graph where the neighborhoods of the vertices of either side are comparable with respect to inclusion.
Let $G_1$ and $G_2$ be two disjoint graphs.
Let $U_i$ be the (possibly empty) set of universal vertices of $G_i$, for $i \in \{1,2\}$.
Let $B$ be any chain graph with bipartition $U_1$, $U_2$.
We call the graph $(G_1 \cup G_2) + E(B)$ a \emph{chain-join} of $G_1$ and $G_2$.
Note that the disjoint union of any two graphs is a particular chain-join of these two graphs.

Regarding our decomposition theorem below, the class of \emph{basic graphs} equals $\{K_1\} \cup \{K_{2n} - ne : n \ge 2\}$.
An equidominating structure of $K_{2n} - ne$, $n \ge 2$, is given by $\vw \equiv 1$ and $t=2$.
Hence, basic graphs are equidominating.
Using Lemma~\ref{lem:universalremoval} below, we see that basic graphs are in fact hereditarily equidominating.
Interestingly, the basic graphs (except for $K_1$) are those graphs which consist of one stable set bundle. 

Let $\mathcal F \coloneqq \{P_5,C_5,{\it bull},{\it banner},{\it house},K_{2,3},\overline{P_2 \cup P_3}\}$ (see Figure~\ref{fig:forbiddensubgraphs} for an illustration).
As the next theorem shows, the set $\mathcal F$ is exactly the set of forbidden induced subgraphs of the class of hereditarily equidominating graphs.

\begin{theorem}\label{thm:hereditarystructure}
For any graph $G$, the following assertions are equivalent.
\begin{enumerate}[label={(\alph*)}]
	\item $G$ is hereditarily equidominating.
	\item $G$ is $\mathcal F$-free.
	\item One of the following assertions holds.
		\begin{enumerate}[label={(\roman*)}]
			\item $G$ is a basic graph.
			\item $G$ is obtained from a hereditarily equidominating graph by adding a universal vertex.
			\item $G$ is the chain-join of two hereditarily equidominating graphs.
		\end{enumerate}
\end{enumerate}
\end{theorem}

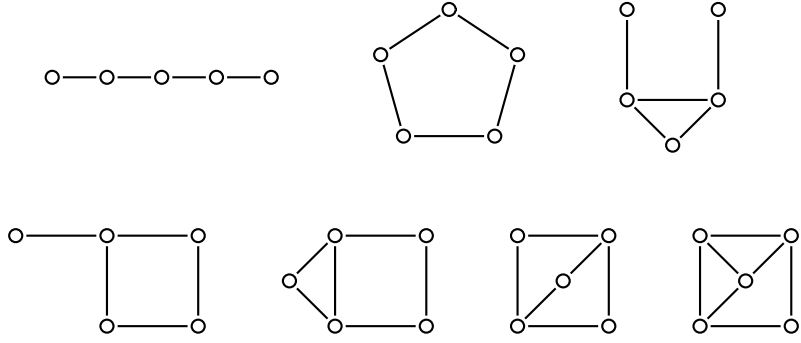
\begin{figure}
\psset{unit=1.2cm,nodesep=0.0333,}
\begin{center}
\begin{pspicture}(0,0)(8.5,4.5)

\cnode(0.4,2.75){0.1cm}{a_1}
\cnode(1,2.75){0.1cm}{a_2}
\cnode(1.6,2.75){0.1cm}{a_3}
\cnode(2.2,2.75){0.1cm}{a_4}
\cnode(2.8,2.75){0.1cm}{a_5}

\ncarc[arcangle=0]{-}{a_1}{a_2}
\ncarc[arcangle=0]{-}{a_2}{a_3}
\ncarc[arcangle=0]{-}{a_3}{a_4}
\ncarc[arcangle=0]{-}{a_4}{a_5}

\cnode(4.25,2.1){0.1cm}{b_1}
\cnode(5.25,2.1){0.1cm}{b_2}
\cnode(5.5,3){0.1cm}{b_3}
\cnode(4.75,3.5){0.1cm}{b_4}
\cnode(4,3){0.1cm}{b_5}

\ncarc[arcangle=0]{-}{b_1}{b_2}
\ncarc[arcangle=0]{-}{b_2}{b_3}
\ncarc[arcangle=0]{-}{b_3}{b_4}
\ncarc[arcangle=0]{-}{b_4}{b_5}
\ncarc[arcangle=0]{-}{b_5}{b_1}

\cnode(7.2,2){0.1cm}{c_1}
\cnode(6.7,2.5){0.1cm}{c_2}
\cnode(7.7,2.5){0.1cm}{c_3}
\cnode(6.7,3.5){0.1cm}{c_4}
\cnode(7.7,3.5){0.1cm}{c_5}

\ncarc[arcangle=0]{-}{c_1}{c_2}
\ncarc[arcangle=0]{-}{c_1}{c_3}
\ncarc[arcangle=0]{-}{c_2}{c_3}
\ncarc[arcangle=0]{-}{c_2}{c_4}
\ncarc[arcangle=0]{-}{c_3}{c_5}

\cnode(0,1){0.1cm}{d_1}
\cnode(1,1){0.1cm}{d_2}
\cnode(2,1){0.1cm}{d_3}
\cnode(1,0){0.1cm}{d_4}
\cnode(2,0){0.1cm}{d_5}

\ncarc[arcangle=0]{-}{d_1}{d_2}
\ncarc[arcangle=0]{-}{d_2}{d_3}
\ncarc[arcangle=0]{-}{d_2}{d_4}
\ncarc[arcangle=0]{-}{d_3}{d_5}
\ncarc[arcangle=0]{-}{d_4}{d_5}

\cnode(3,0.5){0.1cm}{e_1}
\cnode(3.5,1){0.1cm}{e_2}
\cnode(3.5,0){0.1cm}{e_3}
\cnode(4.5,0){0.1cm}{e_4}
\cnode(4.5,1){0.1cm}{e_5}

\ncarc[arcangle=0]{-}{e_1}{e_2}
\ncarc[arcangle=0]{-}{e_1}{e_3}
\ncarc[arcangle=0]{-}{e_2}{e_3}
\ncarc[arcangle=0]{-}{e_2}{e_5}
\ncarc[arcangle=0]{-}{e_3}{e_4}
\ncarc[arcangle=0]{-}{e_4}{e_5}

\cnode(5.5,0){0.1cm}{f_1}
\cnode(5.5,1){0.1cm}{f_2}
\cnode(6,0.5){0.1cm}{f_3}
\cnode(6.5,1){0.1cm}{f_4}
\cnode(6.5,0){0.1cm}{f_5}

\ncarc[arcangle=0]{-}{f_1}{f_2}
\ncarc[arcangle=0]{-}{f_1}{f_3}
\ncarc[arcangle=0]{-}{f_1}{f_5}
\ncarc[arcangle=0]{-}{f_2}{f_4}
\ncarc[arcangle=0]{-}{f_3}{f_4}
\ncarc[arcangle=0]{-}{f_4}{f_5}

\cnode(7.5,0){0.1cm}{g_1}
\cnode(7.5,1){0.1cm}{g_2}
\cnode(8,0.5){0.1cm}{g_3}
\cnode(8.5,1){0.1cm}{g_4}
\cnode(8.5,0){0.1cm}{g_5}

\ncarc[arcangle=0]{-}{g_1}{g_2}
\ncarc[arcangle=0]{-}{g_1}{g_3}
\ncarc[arcangle=0]{-}{g_1}{g_5}
\ncarc[arcangle=0]{-}{g_2}{g_3}
\ncarc[arcangle=0]{-}{g_2}{g_4}
\ncarc[arcangle=0]{-}{g_3}{g_4}
\ncarc[arcangle=0]{-}{g_4}{g_5}

\end{pspicture}
\end{center}
\caption{The forbidden set $\mathcal F$; Top row: $P_5$, $C_5$, {\it bull}; Bottom row: {\it banner}, {\it house}, $K_{2,3}$, $\overline{P_2 \cup P_3}$.}
\label{fig:forbiddensubgraphs}
\end{figure}

The proof of the above theorem builds upon the following lemmas.

\begin{lemma}\label{lem:Fnoted}
Every graph in $\mathcal F$ is not equidominating.
\end{lemma}
\begin{proof}
Let us check, for example, that $P_5$ is not equidominating.
For this, we assume that $V(P_5) = \{1,2,3,4,5\}$ with $E(P_5) = \{12,23,34,45\}$. 

Suppose that $P_5$ is equidominating and let $(\vw,t)$ be an equidominating structure of $P_5$.
As both $\{1,4\}$ and $\{2,4\}$ are mds of $G$, $\vw(1)+\vw(4) = \vw(2)+\vw(4) = t$, and so $\vw(1) = \vw(2)$.
Similarly, $\vw(4) = \vw(5)$.
Hence, $\vw(1) + \vw(5) = \vw(2) + \vw(4) = t$, in contradiction to the fact that $\{1,5\}$ is not a dominating set of $P_5$.
So, $P_5$ is not equidominating.
\end{proof}

\begin{lemma}\label{lem:universalremoval}
Let $G$ be a graph, and let $G'$ be the graph obtained from $G$ by attaching a universal vertex.
Then $G$ is equidominating if and only if $G'$ is equidominating.
\end{lemma}
\begin{proof}
Let $x$ be the universal vertex that is attached to $G$ in order to obtain $G'$.
Let $\mathcal D$ ($\mathcal D'$) be the set of mds of $G$ (of $G'$).
Clearly $\mathcal D' = \mathcal D \cup \{x\}$.

Assume that $G$ is equidominating and let $(\vw,t)$ be an equidominating structure of $G$.
Put $\vw'\big|_{V(G)} :\equiv \vw$ and $\vw'(x)\coloneqq t$.
Then $(\vw',t)$ is an equidominating structure of $G'$.

Now assume that $G'$ is equidominating and let $(\vw',t)$ be an equidominating structure of $G$.
Put $\vw :\equiv \vw'\big|_{V(G)}$.
Then $(\vw,t)$ is an equidominating structure of $G$.
\end{proof}

\begin{lemma}\label{lem:chain-joinpreserves}
Let $G_1$ and $G_2$ be two equidominating graphs, and let $G$ be a chain-join of $G_1$ and $G_2$.
Then $G$ is equidominating.
Moreover, if $G_1$ and $G_2$ are hereditarily equidominating, so is $G$.
\end{lemma}
\begin{proof}
Let $U_i$ be the set of universal vertices of $G_i$, for $i \in \{1,2\}$.
Moreover, let $(\vw_i,t_i)$ be an equidominating structure of $G_i$, for $i \in \{1,2\}$, and let $k = 1 + \sum_{v \in V(G_1)} \vw_1(v)$.
We define a weight function $\vw$ on $G$ by $\vw\big|_{V(G_1)} \equiv \vw_1$ and $\vw\big|_{V(G_2)} \equiv k \cdot \vw_2$.
Put $t \coloneqq t_1 + kt_2$.

By Lemma~\ref{lem:universalremoval}, we may assume that $G$ does not have a universal vertex.
We claim that $(\vw,t)$ is an equidominating structure of $G$.

To see this, pick any set $X \subseteq V(G)$ with $\vw(X) = t$.
By the choice of $\vw$ and $t$, $\vw_1(X \cap V(G_1)) = t_1$ and $\vw_2(X \cap V(G_2)) = t_2$.
Hence, $X_i \coloneqq X \cap V(G_i)$ is an mds of $G_i$, for $i \in \{1,2\}$.
In particular, $X$ is a dominating set of $G$.
We have to show that $X$ is minimal.

Suppose that $X$ is not minimal, that is, there is some $x \in X$ with $N[X \setminus \{x\}] = V(G)$.
W.l.o.g.\ $x \in X_1$.
As $X_1$ is an mds of $G_1$, there is some $y \in V(G_1)$ not dominated by $D_1 \setminus \{x\}$.
Hence, $y \in N_G[X_2]$.
In particular, $y \in U_1$, implying $X_1 \setminus \{x\} = \emptyset$.

So, $N_G[X_2] = V(G)$.
This means $U_2 \neq \emptyset$, and so $N_G[X_2] = N_G[U_2]$.
As the edge cut $E_G(U_1,U_2)$ induces a chain graph, there is some vertex $v \in U_2$ whose set of neighbors in $U_1$ is maximal among the vertices in $U_2$.
Hence, $V(G) = N_G[X_2] = N_G[U_2] = N_G[v]$, and thus $v$ is a universal vertex of $G$, a contradiction.
This shows that $X$ is an mds of $G$.
 
Conversely, let $D$ be an mds of $G$.
We have to prove that $\vw(D)=t$.
It suffices to show that $D_i \coloneqq D \cap V(G_i)$ is an mds $D_i$, for $\{1,2\}$.
Then, $\vw_i(D_i) = t_i$, for $i \in \{1,2\}$, and so $\vw(D) = t_1 + k t_2 = t$.

By symmetry, it suffices to show that $D_1$ is an mds of $G_1$.
As noted above, the assumption that $G$ does not have any universal vertex implies $V(G_1) \not\subseteq N_G[V(G_2)]$.
In particular, $D_1 \neq \emptyset$.
If $U_1 \cap D_1 \neq \emptyset$, it must be that $|D_1|=1$, and thus $D_1$ is an mds of $G_1$.
So we may assume that $U_1 \cap D_1 = \emptyset$.

Since $N_G[V(G_2)] \cap (V(G_1) \setminus U_1) = \emptyset$, $D_1$ is a dominating set of $G_1 - U_1$ and thus of $G_1$.
Suppose that $D_1$ is not minimal, that is, there is some $x \in D_1$ such that $D_1' \coloneqq D_1 \setminus \{x\}$ is still a dominating set of $G_1$.
As $U_1 \cap D_1 = \emptyset$, we have $N_G[D_1] \cap V(G_2) = \emptyset$.
Hence, $D_2$ is a dominating set of $G_2$, and so $D_1' \cup D_2$ is a dominating set of $G$.
This contradicts the minimality of $D$, however.
This completes the proof of the first statement of the lemma.

\bigskip
Now, assume that $G_1$ and $G_2$ are hereditarily equidominating.
Let $H$ be any induced subgraph of $G$.
Observe that $H$ is a chain-join of the two graphs $H_1 \coloneqq G_1[V(H) \cap V(G_1)]$ and $H_2 \coloneqq G_2[V(H) \cap V(G_2)]$.
As $G_1$ and $G_2$ are hereditarily equidominating, $H_1$ and $H_2$ are equidominating.
By applying the first statement of the lemma, we see that $H$ is equidominating, too.
As $H$ is arbitrary, $G$ is hereditarily equidominating.
\end{proof}

We now come to the decomposition of $\mathcal F$-free graphs.

\begin{lemma}\label{lem:dominatingC4}
Let $G$ be a connected $\mathcal F$-free graph without universal vertex.
Assume that $G$ has a minimal connected dominating set $D$ of size 2, say $D =\{x,y\}$, such that there are private neighbors $x'\in pn(x,D)$ and $y'\in pn(y,D)$ with $x'y' \in E(G)$.
Then $G \cong K_{2n} - ne$ for some $n \ge 2$.
\end{lemma}
\begin{proof}
As $G[\{x,y,x',y'\}] \cong C_4$ and $G$ is $\mathcal F$-free, $D' \coloneqq \{x,y,x',y'\}$ is a homogeneous set (that is each vertex of $D'$ sees the same vertices in $V(G)\setminus D'$).
But $D = \{x,y\}$ is a dominating set, and thus also $D'$ is a dominating set.
Hence, $G$ is the complete join of $G[D']$ and $G[V(G)-D']$.

Pick any vertex $z \in V(G)-D'$.
Since $G$ does not have a universal vertex, there must be some non-neighbor $z'$ of $z$.
Suppose that there is a second non-neighbor of $z$, say $z''$.
Consider the graph $G' \coloneqq [\{x,y',z,z',z''\}]$.
If $z'z'' \in E(G)$, $G' \cong \overline{P_2 \cup P_3}$, a contradiction.
Hence, $z'z'' \notin E(G)$, and so $G' \cong K_{2,3}$, again a contradiction.

Thus, every vertex of $G$ has exactly one non-neighbor in $G$.
This means $G \cong K_{2n} - ne$, for some $n \ge 2$.
\end{proof}

\begin{lemma}\label{lem:dominatingbull}
Let $G$ be a connected $\mathcal F$-free graph without universal vertex that admits a connected dominating set of size 2.
Assume that for every minimal connected dominating set $D$ of size 2, say $D =\{x,y\}$, it holds that $x'y' \notin E(G)$ for any two private neighbors $x'\in pn(x,D)$ and $y'\in pn(y,D)$.
Then $G$ is a chain-join of two disjoint connected $\mathcal F$-free graphs $G_1$ and $G_2$.
\end{lemma}
\begin{proof}
Let $\{x,y\}$ be a connected dominating set such that, for all connected dominating sets $\{u,v\}$ of size two of $G$, $|N(u)|+|N(v)| \le |N(x)|+|N(y)|$.
Let $X \coloneqq N[x] \setminus N[y]$, $Y \coloneqq N[y] \setminus N[x]$, and $S \coloneqq N[x] \cap N[y]$.
It is clear that $X$, $Y$, $S$ are pairwise disjoint and $X \cup Y \cup S = V(G)$.
Note that, as $\{x,y\}$ is a minimal connected dominating set, $X,Y \neq \emptyset$.

Let $z \in S \setminus \{x,y\}$ be arbitrary.
Since $G$ does not have a universal vertex, there is a non-neighbor of $z$.
Suppose first that $X \cup Y \subseteq N(z)$.
Then there is some $z' \in S$ with $zz' \notin E(G)$.
Since $G$ is {\it bull}\,-free, $z'$ has some neighbor in $X \cup Y$, say $x' \in X \cap N(z)$.
Pick any $y' \in Y$.
If $y'z' \in E(G)$, then $G[\{x',y,y',z,z'\}] \cong \overline{P_2 \cup P_3}$.
Otherwise, $G[\{x',y,y',z,z'\}] \cong {\it house}$.
Since both is contradictory, we know that $z$ has some non-neighbor among $X \cup Y$, say in $Y$.
As $G$ is {\it bull}-free, $z$ is adjacent to all members of $X$.
In particular, $\{z,y\}$ is a connected dominating set.

Suppose that there is some $y' \in N(z) \cap Y$.
We claim that then $|N(z)| > |N(x)|$, in contradiction to our assumption that, among all connected dominating sets of size two of $G$, $|N(x)|+|N(y)|$ is maximum.
To prove this claim, suppose $|N(z)| \le |N(x)|$.
Since $X \cup \{y'\} \subseteq N(z)$ and $y' \notin N(x)$, there must be some $z' \in S \setminus N(z)$.
Pick any $x' \in X$.
Since $G$ is {\it bull}-free, $x'z' \in E(G)$ or $y'z' \in E(G)$.
W.l.o.g.\ $x'z' \in E(G)$.
Like above, this leads to $G[\{x',y,y',z,z'\}] \cong \overline{P_2 \cup P_3}$ or $G[\{x',y,y',z,z'\}] \cong {\it house}$, depending on whether $y'z' \in E(G)$.
Since both is contradictory, we obtain $N(z) \cap Y = \emptyset$.

Summing up, $X \subseteq N(z)$ and $N(z) \cap Y = \emptyset$.
As $z$ is arbitrary, for every vertex $z' \in S$ it either holds that $X \subseteq N(z')$ and $N(z') \cap Y = \emptyset$ or $Y \subseteq N(z')$ and $N(z') \cap X = \emptyset$.
This partitions the set $S$ into two disjoint sets $X'$ and $Y'$ where every member of $X'$ is adjacent to $X$ and every member of $Y'$ is adjacent to $Y$.
Consider the bipartite graph $B$ with color classes $X'$ and $Y'$ whose edges are given by the edge-cut $F \coloneqq E_G(X',Y')$.
Let $G_1$ ($G_2$) be the connected component of $G - F$ containing $X$ (containing $Y$).
Since the members of $X'$ are universal vertices of $G_1$ and the members of $Y'$ are universal vertices of $G_2$, it remains to prove that $B$ is a chain graph.

Suppose the opposite holds, that is, there are vertices $x',x'' \in X' \cup \{x\}$ and $y',y'' \in Y' \cup \{y\}$ such that $x'y',x''y'' \in E(B)$ and $x'y'',x''y' \notin E(B)$.
Let $x''' \in X$.
Then $G[\{x',x'',x''',y',y''\}] \cong {\it house}$, a contradiction.
This completes the proof.
\end{proof}

\begin{proof}[Proof of Theorem~\ref{thm:hereditarystructure}.]
Assume that $G$ is hereditarily equidominating.
Then by Lemma~\ref{lem:Fnoted} $G$ is $\mathcal F$-free.

\bigskip
Now let $G$ be an $\mathcal F$-free graph.
We may assume that $G$ is not basic and, by Lemma~\ref{lem:universalremoval}, that $G$ does not have universal vertices.
If $G$ is disconnected, it is the chain-join of one of its connected components with the other connected components.
So, we may assume that $G$ is connected.

Let $D$ be any minimal connected dominating set of $G$.
Since $G$ is $(P_5,C_5)$-free and connected, $D$ is a clique~\cite{CS12}.
Suppose that $|D| \ge 3$, and let $x,y,z \in D$ be distinct vertices.
Let $x'$ be a private neighbor of $x$, and $y'$ be a private neighbor of $y$.
If $x'y' \notin E(G)$, $G[\{x,x',y,y',z\}] \cong {\it bull}$, a contradiction.
Thus $x'y' \in E(G)$.
But then $G[\{x,x',y,y',z\}]$ is isomorphic to ${\it house}$, another contradiction.
This shows that $|D|\le 2$.

As we assumed that $G$ does not have a universal vertex, $|D|=2$, say $D=\{x,y\}$.
If $x$ and $y$ have private neighbors, say $x'\in pn(x,D)$ and $y'\in pn(y,D)$, such that $x'y' \in E(G)$, $G$ is basic by Lemma~\ref{lem:dominatingC4}.

So we may assume that for every minimal connected dominating set $D= \{x,y\}$ of $G$, every private neighbor of $x$ is non-adjacent to every private neighbor of $y$.
Hence, $G$ is the chain-join of two $\mathcal F$-free graphs, by Lemma~\ref{lem:dominatingbull}.
\bigskip

Finally, let $G$ be such that one of the following conditions holds.
\begin{enumerate}[label={(\roman*)}]
	\item\label{cond:Gisbasic} $G$ is a basic graph.
	\item\label{cond:universalvertex} $G$ is obtained from a hereditarily equidominating graph by adding a universal vertex.
	\item\label{cond:chain-join} $G$ is the chain-join of two hereditarily equidominating graphs.
\end{enumerate}
If $G$ is basic, it is clearly hereditarily equidominating.
From Lemmas~\ref{lem:universalremoval} and~\ref{lem:chain-joinpreserves} it follows that any graph obtained from a hereditarily equidominating graph by attaching a universal vertex or from two hereditarily equidominating graphs by a chain-join is again hereditarily equidominating.
This completes the proof.
\end{proof}

\subsection{Recognition of hereditarily equidominating graphs}

Given the fact that hereditarily equidominating graphs admit a finite forbidden subgraph characterization, it is clear that this class can be recognized efficiently.
A faster recognition is possible using the decomposition provided by Theorem~\ref{thm:hereditarystructure}.

\begin{corollary}
Let $G$ be a graph on $n$ vertices and $m$ edges.
It can be decided in time $\mathcal O(n (n+m))$ whether $G$ is a hereditarily equidominating graph.
\end{corollary}
\begin{proof}
Let $G$ be a graph on $n$ vertices and $m$ edges.
To decide whether $G$ is hereditarily equidominating, we apply our algorithm presented below to every connected component of $G$.

We use the notion of a \emph{co-chain graph}, which is just the complement of a chain graph.
Equivalently, a co-chain graph is a graph obtained from some chain graph by turning the two bipartition classes into cliques.
These cliques are called the \emph{co-classes} of the co-chain graph.
Note that the partition of a co-chain graph into its two co-classes might not be unique. 

Our algorithm takes as input a connected graph $H$, and mimics the decomposition of Theorem~\ref{thm:hereditarystructure}.
In order to detect whether $H$ is a chain join, the algorithm basically follows the proof of Lemma~\ref{lem:dominatingbull}.
\begin{enumerate}
	\item Compute the degrees of all vertices of $H$.
	\item\label{step:basic} If $H$ is basic, return that $H$ is hereditarily equidominating.
	\item\label{step:universal-vertex} If $H$ has a universal vertex, say $v$, reapply the algorithm to each connected component of $H-v$.
	\item Let $x$ be a vertex of maximum degree, and let $v$ be any non-neighbor of $x$.
	\item Let $y$ be the vertex of maximum degree among $N_H(v)$.
	\item Compute the sets $X \coloneqq N_H(x)\setminus N_H[y]$, $Y \coloneqq N_H(y)\setminus N_H[x]$, and $S \coloneqq N_H[x] \cap N_H[y]$.
	\item Compute the sets $X': = \{x' \in S : N_H(x') \cap X \neq \emptyset\}$ and $Y' \coloneqq S \setminus X'$.
	\item Compute the graph $H' \coloneqq H[S]$.
	\item\label{step:co-chain-recognition} Check whether $H'$ is a co-chain graph with co-classes $X'$ and $Y'$. If not, return that $G$ is not hereditarily equidominating.
	\item\label{step:universal-conditions} Check whether the following conditions are satisfied: 
		\begin{itemize}
			\item for every $x' \in X'$ it holds that $X \subseteq N_H(x')$ and $Y \cap N_H(x') = \emptyset$;
			\item for every $y' \in Y'$ it holds that $Y \subseteq N_H(y')$ and $X \cap N_H(y') = \emptyset$.
		\end{itemize}	
		If one of these conditions fails to hold, return that $G$ is not hereditarily equidominating.
	\item\label{step:final-reapply} Reapply the algorithm to each connected component of $H-S$.
\end{enumerate}

Let us first show how the algorithm can be implemented such that each iteration runs in $\mathcal O (|V(H)| + |E(H)|)$ time.
Since at least one vertex is removed in each iteration, the overall running time of the algorithm then computes to $\mathcal O(n (n+m))$.

Note that $H \cong K_{2k}-ke$, for some $k \ge 2$, if and only if $|V(H)| \ge 4$ and every vertex of $H$ has degree $|V(H)|-2$.
Hence, it can be checked in linear time whether $H$ is basic.
Since all other steps are standard, it remains to discuss how to perform Step~\ref{step:co-chain-recognition}.
In this particular step, we have to decide whether $H'$ is a co-chain graph with the prescribed co-classes $X'$ and $Y'$.
As shown by Heggernes and Kratsch~\cite{HK07}, co-chain graphs can be recognized in time $\mathcal O (|V(H)| + |E(H)|)$.
It is straightforward that the algorithm of Heggernes and Kratsch can be modified such that it includes prescribed co-classes.
Hence, Step~\ref{step:co-chain-recognition} can be performed in time $\mathcal O (|V(H)| + |E(H)|)$.
\medskip

We now come to the correctness of the algorithm.
Our aim is to show that the algorithm performs a decomposition according to Theorem~\ref{thm:hereditarystructure} (if possible).
For this, it suffices to show that in each iteration, the algorithm performs a single step of such a decomposition, or correctly decides that a decomposition is no longer possible.

Let us first show that if $H$ is a hereditarily equidominating graph, the algorithm correctly decomposes $H$.
By Theorem~\ref{thm:hereditarystructure}, $H$ is either basic, has a dominating vertex or is a chain-join.

In the case that $H$ is basic, which is checked in Step~\ref{step:basic}, the algorithm can safely return that $H$ is hereditarily equidominating.
If $H$ has a dominating vertex, say $v$, this is detected in Step~\ref{step:universal-vertex}.
By Theorem~\ref{thm:hereditarystructure}, $H$ is hereditarily equidominating if and only if $H-v$ is hereditarily equidominating.
Hence, the algorithm is correctly reapplied to the connected components of $H-v$.

So let us assume that $H$ is neither basic nor has a universal vertex.
Then, according to Theorem~\ref{thm:hereditarystructure}, $H$ is the chain-join of two graphs $H_1,H_2$.
Let $U_i$ be the set of universal vertices of $H_i$, for $i \in \{1,2\}$.
Let $x$ be a vertex of maximum degree in $H$.
As $H$ is a chain-join of $H_1$ and $H_2$, it holds that $x \in U_1 \cup U_2$, say $x \in U_1$.
Moreover, $x$ has the maximal closed neighborhood in $H$ among all vertices of $U_1$.

Let $v \in V(H) \setminus N_H[x]$, and let $y$ be the vertex of maximum degree among $N_H(v)$.
Then $v \in V(H_2) \setminus N_H(x) = V(H_2) \setminus N_H(V(H_1))$.
Hence, $y \in U_2$, and $y$ has the maximal closed neighborhood in $H$ among all vertices of $U_2$.

Now we define $X \coloneqq N_H(x)\setminus N_H[y]$, $Y \coloneqq N_H(y)\setminus N_H[x]$, $S \coloneqq N_H[x] \cap N_H[y]$, $X': = \{x' \in S : N_H(x') \cap X \neq \emptyset\}$, and $Y' \coloneqq S \setminus X'$.
Since $X' \subseteq U_1$, $Y' \subseteq U_2$, and $H$ is a chain-join, it holds that $H[S]$ is a co-chain graph with co-classes $X'$ and $Y'$.
Thus, both $X'$ and $Y'$ satisfy the conditions checked in Step~\ref{step:universal-conditions}.
Note that the graph $H-S$ equals the disjoint union of $H_1-X'$ and $H_2-Y'$, both being hereditarily equidominating by assumption.
Hence, the algorithm is correctly reapplied to the connected components of $H-S$.

Let us now assume that the algorithm performs one iteration on $H$ and does not return that $G$ is not hereditarily equidominating.
We have to show that $H$ is hereditarily equidominating if and only if all graphs are hereditarily equidominating to which the algorithm is reapplied.
Clearly this holds if $H$ has a universal vertex, so we proceed to analyze the case that the reapplication is called in Step~\ref{step:final-reapply}.

In this case the sets $S$, $X'$, $Y'$, and the graph $H'$ are computed, and both Step~\ref{step:co-chain-recognition} and Step~\ref{step:universal-conditions} are performed successfully.
In particular, $H'$ is a co-chain graph with co-classes $X'$ and $Y'$.
This means that $X'$ and $Y'$ are cliques in $H'$ and thus in $H$, and that the graph induced by the edges $E_H(X',Y')$ is a chain graph.
After Step~\ref{step:universal-conditions}, we know that for every $x' \in X'$ it holds that $X \subseteq N_H(x')$ and $Y \cap N_H(x') = \emptyset$, and for every $y' \in Y'$ it holds that $Y \subseteq N_H(y')$ and $X \cap N_H(y') = \emptyset$.
Hence, $H$ is a chain-join of the two connected components of the graph $H-E_H(X',Y')$.
Let us denote these two components by $H_1$ and $H_2$.
Moreover, let $H_1'\coloneqq H_1-X'$ and $H_2'\coloneqq H_2-Y'$.
Note that $H-S = H_1' \cup H_2'$.

By Theorem~\ref{thm:hereditarystructure}, $H$ is hereditarily equidominating if and only if both $H_1$ and $H_2$ are hereditarily equidominating.
Since $X'$ is a set of universal vertices of $H_1$, $H_1$ is hereditarily equidominating if and only if $H_1'$ is hereditarily equidominating, again by Theorem~\ref{thm:hereditarystructure}.
The analogous statement holds for $H_2$ and $H_2'$.
Summing up, $H$ is hereditarily equidominating if and only if $H-S = H_1' \cup H_2'$ is hereditarily equidominating.
This completes the proof.
\end{proof}

\section{Conclusion and Outlook}\label{sec:conclusion}
The main result of this paper is that the {\sc Equidomination} problem can be parameterized in two different ways such that the parameterized problems are fixed-parameter tractable.
One way is to take the target value $t$ of an equidominating structure as the parameter: this problem is called the {\sc Target}-$t$ {\sc Equidomination} problem.
The second way is to allow only vertex weights up to a certain value $k$, which leads to the $k$-{\sc Equidomination} problem.
%Furthermore we gave two algorithms to solve both problems.

To solve the kernelized instances we developed an XP algorithm that can be used for both problems.
Even though there are many analogies between equistability and equidomination ---as between the concepts of stability and domination in general--- there are new difficulties in the case of equidomination.
This fact results in several extensions we needed to discover and formalize.
Finally, we characterized hereditarily equidominating graphs in terms of seven forbidden subgraphs, which also leads to a polynomial time recognition algorithm.

There are several open questions that arise in the context of equidomination. 
Most importantly, we would like to see a hardness proof of the {\sc Equidomination} problem.
For this, however, it might be necessary to get a better grip on the combinatorial properties of equidominating graphs.
One way in this direction could be a characterization of target-$t$ equidominating and $k$-equidominating graphs for small $t,\, k\in \IN$.

Further, one could adapt the idea of equidomination ---namely characterizing mds by the $0$-$1$-solutions of linear equalities--- to other forms of domination like total, multiple or global domination.
Thoughts in this direction could concern a characterization of these graph classes as well as the complexity of the according decision problems. 

%\bibliography{../bibliography}

\end{document}